\documentclass{article}
\usepackage[margin=1in]{geometry}

\usepackage{amsthm, amsfonts}
\usepackage{amsmath}
\usepackage{comment}
\usepackage{mathrsfs}
\usepackage{caption}
\usepackage{enumitem}
\usepackage{marvosym}
\usepackage{tikz}
\usepackage[
backend=biber,
style=alphabetic,
]{biblatex}

\usepackage{graphicx}
\usepackage{xcolor}

\usepackage{epsfig}
\usepackage{graphicx,amssymb,textcomp}
\graphicspath{ {Images/} }

\theoremstyle{plain}

\newtheorem{lemma}{Lemma}[section]
\newtheorem{cor}{Corollary}[section]
\newtheorem{theorem}{Theorem}[section]
\newtheorem{prop}{Proposition}[section]

\theoremstyle{remark}
\newtheorem{rem}{Remark}[section]

\numberwithin{equation}{subsection}

\let\oldsection\section
\renewcommand{\section}{
  \renewcommand{\theequation}{\thesection.\arabic{equation}}
  \oldsection}
\let\oldsubsection\subsection
\renewcommand{\subsection}{
  \renewcommand{\theequation}{\thesubsection.\arabic{equation}}
  \oldsubsection}

\newcommand{\IR}{\mathbb{R}}

\newcommand{\IS}{\mathbb{S}}

\newcommand{\mf}{\mathcal{F}}

\newcommand{\mv}{\mathcal{V}}

\renewcommand{\d}{\mathop{}\!\mathrm{d}}

\renewcommand{\bar}{\overline}

\newcommand{\bracket}[1]{\left\langle#1\right\rangle}

\newcommand{\norm}[1]{\left\lVert#1\right\rVert}

\newcommand{\tr}{\text{tr}}
\newcommand{\vol}{\text{Vol}}

\newcommand{\pd}{\partial}

\makeatletter

\newcommand{\charfusion}[3][\mathord]{
	#1{\ifx#1\mathop\vphantom{#2}\fi
		\mathpalette\mov@rlay{#2\cr#3}
	}
	\ifx#1\mathop\expandafter\displaylimits\fi}
\makeatother

\newcommand{\bg}{\bar{g}}

 \newcommand{\sch}{\mbox{Sch}}
 
 \newcommand{\uep}{U_{a(\epsilon_1)}^A}
\newcommand{\wuep}{\widetilde{U}_{a(\epsilon_1)}^A}
\newcommand{\upa}{U_{\rho_0}^A}

\newcommand{\spa}{\; | \;}

\title{Stability of the Spacetime Penrose Inequality in Spherical Symmetry}
\author{Emily Schaal}
\bibliography{bibliography} 

\begin{document}
\maketitle

\section*{Abstract}

We formulate and prove the stability statement associated with the spacetime Penrose inequality for $n$-dimensional spherically symmetric, asymptotically flat initial data satisfying the dominant energy condition. We assume that the ADM mass is close to the half area radius of the outermost apparent horizon and, following the generalized Jang equation approach, show that the initial data must arise from an isometric embedding into a static spacetime close to to the exterior region of a Schwarzschild spacetime in the following sense. Namely, the time slice is close to the Schwarzschild time slice in the volume preserving intrinsic flat distance, the static potentials are close in $L_{loc}^2$, and the initial data extrinsic curvature is close to the second fundamental form of the embedding in $L^2$.

\section{Introduction}

The Penrose inequality was first conjectured by Roger Penrose in the 1970s. It bounds from below the ADM mass of an asymptotically flat initial data set for the Einstein equations by the smallest area enclosing its trapped surfaces. In other words, a spacetime with a black hole must have at least as much mass as the half area radius of the black hole. The rigidity statement for the Penrose inequality posits that if equality is achieved, then the spacetime must actually be Schwarzschild. We explore the stability conjecture, which states that if the mass is close to the half area radius, then the spacetime should be close to Schwarzschild, and find that stability holds in the volume preserving intrinsic flat sense.

The Penrose inequality in the 3-dimensional Riemannian case was solved by Huisken and Ilmanen in 2001 for a manifold with one boundary component using inverse mean curvature flow \cite{Huisken2001}, and shortly thereafter by Bray for general boundary using conformal flow \cite{Bray2001}. Later, Bray and Lee generalized the conformal flow argument to dimensions less than $8$ \cite{Bray2009}. On the other hand, the spacetime Penrose inequality has been fully proven only in the case of spherical symmetry. In 1996, Hayward proved the inequality for spherically symmetric initial data but did not handle the case of equality \cite{Hayward1996}, while Bray and Khuri published a proof of the inequality and equality case for spherically symmetric initial data using a type of generalized Jang equation in 2010 \cite{Bray2010}. Although their proof is only explicitly stated in dimension 3, it can be generalized to all dimensions due to spherical symmetry.

Both the Riemannian and spacetime Penrose inequalities have stability conjectures corresponding to the case of equality. As described previously, such conjectures ask whether the manifold or the associated static spacetime is close to Schwarzschild when the ADM mass is close to the half area radius of the black hole. 

We appeal to the intrinsic flat distance to judge whether two manifolds are close together. The intrinsic flat distance was established by Sormani and Wenger in \cite{Sormani2011} as a generalization of the flat distance in a manner analogous to how the Gromov-Hausdorff distance generalizes the Hausdorff distance; instead of computing the flat distance between integral currents in Euclidean space, the intrinsic flat distance is taken as an infimum over all distance-preserving embeddings into all complete metric spaces. The weak nature of the intrinsic flat distance makes it ideal for problems like the Penrose inequality, which rely on assumptions regarding scalar curvature as opposed to stronger conditions regarding Ricci curvature. 

Bryden, Khuri, and Sormani used the Lakzian-Sormani technique \cite{Lakzian2012} to show the intrinsic flat distance stability of the spacetime positive mass theorem in the case of spherical symmetry \cite{Bryden_2020}. However, the Penrose inequality, for which the boundary of the initial data remains nontrivial, poses complications when it comes to establishing stability. In 2011, Lee and Sormani proved a type of stability result for the Riemannian Penrose inequality in spherical symmetry using the Lipschitz distance, which in turn implied intrinsic flat closeness \cite{Lee2011}. In their paper, Lee and Sormani highlight that a pure stability result is impossible in the Riemannian case because the control over the boundary does not prohibit the manifold from developing a cylindrical end of arbitrary length. Instead, they show that if the ADM mass is close to the half area radius, then the manifold is close to a so-called appended Schwarzschild space constructed by attaching a cylinder to the boundary of the appropriate Schwarzschild time slice. Moreover, in \cite{Mantoulidis_2015}, Mantoulidis and Schoen published a counterexample to the stability of the Riemannian Penrose inequality by gluing an almost-cylinder with a poorly behaved minimal boundary onto a portion of the exterior region of Schwarzschild space. While their resulting manifold is not spherically symmetric, it provides further evidence that the behavior near the boundary cannot quite be controlled by the assumptions regarding scalar curvature and mass. Additionally, Sakovich and Sormani have proven a stability result for the spacetime positive mass theorem in the asymptotically hyperbolic setting \cite{Sakovich2017}. In \cite{Sakovich_2021}, Sakovich solves the Jang equation for asymptotically hyperboloidal initial data in order give a non-spinor proof of the hyperbolic positive mass theorem. Additionally, Allen studies stability of the 3-dimensional Riemannian positive mass theorem and Penrose inequality in a number of cases using inverse mean curvature flow in \cite{ Allen2017b, Allen2017,Allen2018a}.

In this paper, we prove a stability result for the spacetime Penrose inequality in spherical symmetry with asymptotically flat initial data using the generalized Jang equation approach. While we are able to adapt the Jang equation proof of the Penrose inequality presented by Bray and Khuri \cite{Bray2010, Bray_2011} to the stability case, we find that an additional assumption is needed regarding the mean curvature of spheres in order to properly control the geometry outside of the asymptotic end. This assumption is closely tied to the outermost apparent horizon condition, and while technical in nature, is indispensable for the argument. Additionally, we must appropriately address the difficulties presented by lack of control near the boundary, which, similar to the case presented by Lee and Sormani, could develop into a cylinder of positive length. Finally, unlike \cite{Bryden_2020}, we cannot use the Lakzian-Sormani technique to show intrinsic flat closeness. The reason is that our initial data might have countably many regions in which we cannot guarantee metric closeness and therefore we need a technique that is more robust to measure the boundary of the regions on which we are evaluating intrinsic flat closeness. Instead, we apply the new boundary version of the volume above, distance below (VADB) estimation theorem \cite{Allen2020a}. 

The paper is structured as follows. In Section \ref{setup}, we rigorously state the setup for the problem and state the main theorem. In Section \ref{background}, we review some of the necessary background for the proof of the theorem, including the volume above, distance below theorem used to estimate intrinsic flat convergence as well as the proof of the Penrose inequality via the generalized Jang method. In Section \ref{prelims}, we prove some preliminary results that are used throughout. In Section \ref{warp fact convergence}, we prove the convergence result for the warping factor and consequently show that the remaining inner region of the comparison Schwarzschild space converges to the zero current in the intrinsic flat sense. In Section \ref{vadb app}, we establish the results which allow us to show intrinsic flat closeness. In Section \ref{proof}, we complete the proof of the main theorem.\bigskip

\noindent \textbf{Acknowledgements.}
The author would like to thank Professor Marcus Khuri for his invaluable advice and encouragement. The author would also like to thank Professor Christina Sormani for several helpful and stimulating conversations. 

\section{Setup and Theorem Statement}\label{setup}

We begin with an initial data set for the Einstein equations $(M,g,k)$ in which $M$ is a complete $n$-dimensional smooth manifold with Riemannian metric $g$, $n \geq 3$. The tensor $k$ is a symmetric 2-tensor which we treat as a candidate for the second fundamental form of an embedding of $M$ into a spacetime. The local matter $\mu$ and momentum density $J$ satisfy the constraint equations given by
\begin{align}
    16\pi \mu &= R_g - (\tr_g k)^2 - |k|_g^2, \quad   8\pi J = \mbox{div}_g(k - (\tr_g k)g).
\end{align}
We assume the dominant energy condition, ie,
\begin{align}\label{assumption: DEC}
    \mu \geq |J|_g.
\end{align}

We assume that $(M,g,k)$ is asymptotically flat in the sense that there is an asymptotic end that is diffeomorphic to the complement of a ball $\IR^n \setminus B_0(R)$ and there exists a uniform constant $C$ so that in the coordinates $x$ provided by this asymptotic diffeomorphism we have the following fall off conditions
\begin{equation}
\begin{aligned}\label{ass: asym flat}
    \left| \pd^{\beta_1}(g_{ij} - \delta_{ij}) \right| &\leq \frac{C}{|x|^{n-2+|\beta_1|}},&  \left| \pd^{\beta_2}k_{ij} \right| &\leq \frac{C}{|x|^{n-1+|\beta_2|}},\\
    |R_g| &\leq \frac{C}{|x|^{n+1}},&  |\tr_g k| &\leq \frac{C}{|x|^n}
\end{aligned}
\end{equation}
for multi-indices $\beta_1 \leq 2$, $\beta_2 \leq 1$. This assumption is necessary to ensure that the ADM mass is well defined as a limit of the Hawking mass of compact hypersurfaces at infinity. 

We restrict to spherically symmetric initial data: ie, $M$ has symmetry group $\mbox{SO}(n)$. In spherical symmetry, we can write the metric $g$ globally as
\begin{align}
    g = g_{11}(r) dr^2 + \rho^2(r) d\Omega^2
\end{align} 
for radial functions $g_{11}$ and $\rho$ where $d\Omega^2$ is the round metric on $\IS^{n-1}$. If $\eta$ is the outer unit normal vector to the sphere of radius $r$, we can write $k$ as
\begin{align}
    k_{ij} = \eta_i \eta_j k_n(r) + (g_{ij} - \eta_i \eta _j) k_t(r)
\end{align}
for radial functions $k_n$ and $k_t$, where $k_n$ represents the component of $k$ normal to spheres and $k_t$ the tangent component.  

We assume that the null expansions satisfy
\begin{align}\label{assumption: apparent horizon}
	\Theta_{\pm}:=2\left(\frac{\rho_{,r}}{\rho\sqrt{g_{11}}}\pm k_n \right)(r)>0
\end{align}
and call this the apparent horizon condition, where $\rho_{,r}$ denotes the derivative of $\rho$ with respect to $r$. While this was sufficient to prove the Penrose inequality and rigidity case, we need a slightly stronger assumption for stability, namely, that the null expansions are uniformly bounded:
\begin{align}\label{assumption: bounded expansions.}
    0 < \Theta_{\pm}(r) \leq C < \infty.
\end{align}
Adding $\Theta_{+}$ and $\Theta_{-}$ gives that $\rho_{,r} > 0$ so that $\rho$ is strictly increasing. We may also see that 
\begin{align}
    0 < H_{S_r} = 2\frac{\rho_{,r}}{\rho\sqrt{g_{11}}} \leq C
\end{align}
where $H_{S_r}$ is the mean curvature to spheres in $g$. The asymptotics \eqref{ass: asym flat} guarantee that the null expansions decay at infinity, so the uniform boundedness condition prevents singularites in the radial length from appearing in the interior. 

In spherical symmetry we only have one boundary component so that the half area radius may be computed as  
\begin{align}\label{eq: def half area rad}
    m_0 = \frac{1}{2}\left( \frac{|\pd M|}{\omega_{n-1}} \right)^{\frac{n-2}{n-1}}
\end{align}
where $|\pd M|$ denotes the area of sphere which composes the boundary. In arc length coordinates, we may define the ADM mass as the limit of the Hawking mass as $s\to \infty$ of spheres of radius $s$, denoted by $S_s$
\begin{align}
   m_{ADM} = \lim_{s\to \infty} m(s) &= \lim_{s \to \infty} \frac{1}{2}\left(\frac{|S_s|}{\omega_{n-1}} \right)^{\frac{n-2}{n-1}}\left[1 - \frac{1}{(n-1)^2(\omega_{n-1}^2 |S_s|^{n-3})^{\frac{1}{n-2}}}\int_{S_s} H_{S_s}^2 \right]
\end{align}
where $\omega_{n-1}$ is the area of the $(n-1)$-sphere and $H_{S_s}$ is the mean curvature of the sphere of radius $s$ in $M$.

The spacetime Penrose inequality then states that, given the dominant energy condition \eqref{assumption: DEC}, the asymptotic flatness given by \eqref{ass: asym flat}, and the outermost apparent horizon condition \eqref{assumption: apparent horizon}, we have
\begin{align}
    m_{ADM} \geq m_0.
\end{align}
Moreover, if $m_{ADM}= m_0$, then $(M,g,k)$ arises as an embedded timeslice of Schwarzschild spacetime with second fundamental form $k$.

To show stability, we follow the generalized Jang approach used to prove the Penrose inequality in \cite{Bray2010, Bray_2011}. The philosophy is to look for a surface $(\Sigma,\bg)$ inside the warped product space $(M\times \IR, g + \phi^2 dt^2)$ such that $\Sigma = F(M)$ for an embedding 
\begin{align}\label{jang as image under f}
    F: M \to M\times \IR, \quad F(x) = (x,f(x))
\end{align}
which satisfies 
\begin{align}\label{eq: Jang Equation}
    H_{\Sigma} - \tr_{\Sigma} K = 0
\end{align}
where $K$ is an extension of the initial data $k$ given by
\begin{equation}\label{eq: extension for K}
K_{ij} = \begin{cases}
    k_{ij} &\mbox{ for } i,j \leq n\\
     0 &\mbox{ for } i = n+1, j \ne n +1\\
    \bracket{N, \phi \nabla_{g}\phi }_{g + \phi^2 dt^2} &\mbox{ for } i = j = n+1
\end{cases}
\end{equation}
and $N$ is the normal vector to $\Sigma$ inside $M\times \IR$. We denote the second fundamental form of the Jang surface $(\Sigma,\bg)$ by $h$. The choice of extension for $k$ to $K$ in \eqref{eq: extension for K} is made to give a positivity property for the scalar curvature of $\bg$ when we assume the dominant energy condition.

Once we have the embedding $F$ to specify $\Sigma$, we can define $G: (M,g) \to (\Sigma \times \IR, \bg - \phi^2 dt^2)$ as the graph map $G(x) = (x, f(x))$ so that the induced metric on $G(M)$ is $g = \bg - \phi^2 df^2$. To be consistent with \cite{Bray2010}, let $\pi$ be the second fundamental form of $G(M)$. If $\pi - k = 0$, then $k$ is the second fundamental form of the embedding. If $\bg = g_S$ is the Schwarzschild metric and $\phi = \phi_S$, then the map $G$ gives $(M,g)$ as the image of an isometric embedding of $\Sigma$ into Schwarzschild spacetime. The existence of the embedding and the equality case were proven in \cite{Bray2010}. For stability, we want to show that if $(M,g)$ is close to $(\sch(m_0),g_S)$, $\phi$ is close to $\phi_S$ and $k$ is close to $\pi$ so that $(M,g,k)$ is close to being realized as an embedded timeslice of Schwarzschild spacetime where we denote the Schwarzschild space with half area radius $m_0$ and metric $g_S$ by $(\sch(m_0),g_S)$..

Note that we will refer to $f$ and $v$ both as solutions to the Jang equation where $v$ is given in $r$ coordinates by
\begin{align}\label{eq: eq for v}
    v = \frac{\phi \sqrt{g^{11}}f_{,r}}{\sqrt{1+ \phi^2 g^{11}f_{,r}^2}},
\end{align}
see \cite{Bray2010}. In his 1979 paper \cite{Jang1977}, Pong Soo Jang proved the equivalence of the existence and uniqueness of $v$ to the existence and uniqueness of the graph determined by $f$, so we may solve for $v$ in \eqref{eq: Jang Equation} to obtain the embedding. Moreover, we choose the boundary data for the Jang solution to preserve the half area radius $m_0$.

While we may state the generalized Jang equations for arbitrary warping factor $\phi$, the existence of the Jang surface has only been proven for a specific choice. Let
\begin{align}
    s = \int_0^r \sqrt{g_{11} + \phi^2 f_{,r}^2}
\end{align}
be the radial arc length parameter in the $\bg$ metric. In these coordinates, 
\begin{align}
    \bg = ds^2 + \rho^2(s) d\Omega^2
\end{align}
and the mean curvature of a sphere of radius $s$ is 
\begin{align}\label{eq: mean curv sphere}
    \bar{H}_{S_s} = \frac{n-1}{\rho}\rho_{,s}.
\end{align}
We set the warping factor $\phi$ as
\begin{align}\label{eq: set phi}
    \phi = \rho_{,s} = \frac{\bar{H}_{S_s}}{n-1}  \left(\frac{|S_s|}{\omega_{n-1}}\right)^{1/(n-1)}. 
\end{align}
This choice of $\phi$ couples the Jang equation to the inverse mean curvature flow, which is solvable in spherical symmetry. The solvability of the flow then produces a solvable Jang system. In a sense, this technique is a combination of the Huisken and Ilmanen approach with the original approach of Schoen and Yau to the spacetime positive mass theorem \cite{Schoen_1981}. 

With the Jang setup in place, we move on to defining the sets on which we study intrinsic flat convergence. In radial coordinates, the metric becomes
\begin{align}
    \bg = \frac{1}{\rho^2_{,s}} d\rho^2 + \rho^2 d\Omega^2.
\end{align}

\begin{rem}
That $\rho$ is strictly increasing implies that, in these coordinates, there is a bijection between the set of radial values and the set of areas, so distinguishing a spherically symmetric subset by its radial values is the same as distinguishing it by the areas of the spheres it contains: the unique sphere of radius $\rho$ has area $\rho^{n-1} \omega_{n-1}$. In this way, we see that these coordinates have geometric meaning. 
\end{rem}

We wish to isolate the minimal radial value. Again, as $\rho$ is strictly increasing, this will be the radial value corresponding to the area of the innermost boundary. As we choose Jang data to preserve the half area radius $m_0$, we compute that in spherical symmetry the minimal radial value is
\begin{align}
    \rho_0 = (2m_0)^{1/(n-2)}.
\end{align}
It follows that a point $x \in \Sigma$ can be written as $x = (\rho, \theta)$ for $\rho \in [\rho_0, \infty)$ and $\theta \in \IS^{n-1}$. We compare to Schwarzschild space with the same half area radius $m_0$, given in radial coordinates by
\begin{align}
    (\sch(m_0),g_S) = \left([\rho_0,\infty] \times \IS^{n-1}, \frac{1}{\phi_S^2} d\rho^2 + \rho^2 d\Omega^2 \right)
\end{align}
where $\phi_S = \sqrt{1 - \frac{2m_0}{\rho^{n-2}}}$, $d\Omega^2$ is the standard metric on the $(n-1)$-sphere, and $\rho_0^{n-2} = 2m_0$ and see that we have a natural diffeomorphism $F_1: \Sigma \to \sch(m_0)$ so that $(\rho,\theta) \mapsto (\rho,\theta)$.

We distinguish spheres and annuli in $(\Sigma,\bg)$ and $(\sch(m_0),g_S)$ with the convention that if $S$ is a set in $\Sigma$, $\widetilde{S}$ is a set in $\sch(m_0)$:
\begin{equation}
\begin{aligned}
	S_A &= \{(\rho,\theta)\in \Sigma \: : \: \rho = A \}, &
	U_a^A &= \{(\rho, \theta) \in \Sigma \: :\: a \leq \rho \leq A \},\\
	\widetilde{S}_A &= \{(\rho,\theta) \in \sch(m_0) \: : \: \rho = A \},&
	\widetilde{U}_a^A  &= \{(\rho, \theta) \in \sch(m_0) \: :\: a \leq \rho \leq A \},
\end{aligned}
\end{equation}
and let
\begin{align}
    T_D(S_A) = \{ x \in \Sigma \spa d_{\bg}(x,y) \leq D \mbox{ for } y \in S_A \}
\end{align}
be the tubular neighborhood of radius $D$ around $S_A$.

We choose $A >> \rho_0$. Given the anchor surface $S_A$, we select the annulus on which we prove intrinsic flat convergence. To do so, we must be careful of how radially long our annulus is. We must also take care that the inner boundary of our annulus is not too close to the boundary of the Jang surface. We first control for the radial length by choosing the distance from which we measure down on the anchor surface $S_A$. Let
\begin{align}
    \widetilde{D} = \min\left\{d_{g_S}(x,y) \spa x \in \widetilde{S}_A, y \in \pd\sch(m_0) = \widetilde{S}_{\rho_0} \right \}
\end{align}
be the distance between the anchor surface in $\sch(m_0)$ and $\pd \sch(m_0)$. We set
\begin{align}
    a(\epsilon_1) := \min\{\rho \spa (\rho,\theta) \in T_{\widetilde{D}}(S_A) \} + \epsilon_1
\end{align}
for a fixed $\epsilon_1 > 0$. The constant $\epsilon_1$ controls for the case that the radial distance between $S_A$ and $\pd \Sigma$ is $\widetilde{D}$ or less. Specifically, we are preventing $T_{\widetilde{D}}(S_A) \cap U_{\rho_0}^A = U_{\rho_0}^A$: by adding $\epsilon_1$ we guarantee that the inclusion $U_{a(\epsilon_1)}^A \subset U_{\rho_0}^A$ is strict. See Figure \ref{fig:set diagram} for an illustration of $\uep$ in $\Sigma$.

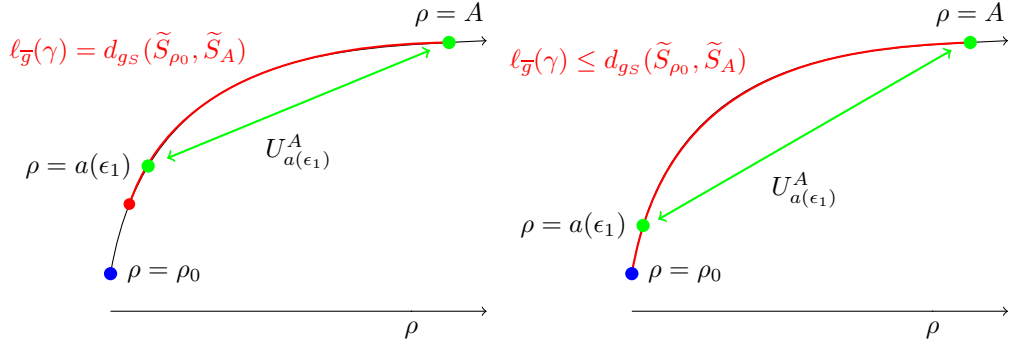
\begin{figure}[ht]
\centering
\begin{tikzpicture}

\draw[->] (0,-1) .. controls (0.5,2) and (3,2) .. (5,2.1);
\node[circle, inner sep=0pt, minimum size = 5pt, fill=blue, label={0: $\rho = \rho_0$}] at (0,-1) {};
\filldraw [red] (0.25,-0.07437) circle (2pt);

\draw[red, thick](0.25,-0.07437)  .. node[above left]{$\ell_{\bg}(\gamma) = d_{g_S}(\widetilde{S}_{\rho_0},\widetilde{S}_{A})$} controls (0.85,1.65) and (2.75,2.05) .. (4.5,2.07574) ;

\node[circle, inner sep=0pt, minimum size = 5pt, fill=green, label={90: $\rho = A$}] at (4.5,2.07574) {};
\node[circle, inner sep=0pt, minimum size = 5pt, fill=green, label={180: $\rho = a(\epsilon_1)$}] at (0.5,0.432282) {};
\draw[green, thick, <->](0.75,0.534999)  -- (4.25,1.97303) ;
\node[circle, inner sep=0pt, minimum size = 1pt, fill=white,label={270:$U_{a(\epsilon_1)}^A$}] at (2.5,1.0) {};
\draw[->] (0,-1.5) -- (5,-1.5);
\node[circle, inner sep=0pt, minimum size = 0.5pt, fill=black, label={270: $\rho$}] at (4,-1.5) {};
\end{tikzpicture}
\begin{tikzpicture}

\draw[->] (0,-1) .. controls (0.5,2) and (3,2) .. (5,2.1);

\draw[red, thick](0,-1)  .. node[above left]{$\ell_{\bg}(\gamma) \leq d_{g_S}(\widetilde{S}_{\rho_0},\widetilde{S}_{A})$} controls (0.45,1.65) and (2.5,2) .. (4.5,2.07574) ;

\node[circle, inner sep=0pt, minimum size = 5pt, fill=blue, label={0: $\rho = \rho_0$}] at (0,-1) {};
\node[circle, inner sep=0pt, minimum size = 5pt, fill=green, label={90: $\rho = A$}] at (4.5,2.07574) {};
\node[circle, inner sep=0pt, minimum size = 5pt, fill=green, label={180: $\rho = a(\epsilon_1)$}] at (0.15,-0.358729) {};

\draw[green, thick, <->](0.35,-0.279786)  -- (4.25,1.97303) ;
\node[circle, inner sep=0pt, minimum size = 1pt, fill=white,label={270:$U_{a(\epsilon_1)}^A$}] at (2.3,0.489908) {};
\draw[->] (0,-1.5) -- (5,-1.5);
\node[circle, inner sep=0pt, minimum size = 0.5pt, fill=black, label={270: $\rho$}] at (4,-1.5) {};
\end{tikzpicture}
\caption{Here we have sketched radial cross-sections of $\Sigma$ illustrating the two different cases that might occur when choosing the set $\uep$. The first, pictured left, is that the radial distance between $S_{\rho_0}$ and $S_{A}$ is greater than $d_{g_S}(\widetilde{S}_{\rho_0},\widetilde{S}_{A})$. In this case, the radial geodesic $\gamma$ that drops down toward $S_{\rho_0}$ and has length $\ell_{\bg}(\gamma)  = d_{g_S}(\widetilde{S}_{\rho_0},\widetilde{S}_{A}) $ will terminate at some radial value bigger than $\rho_0$, depicted by the red dot. The set $U_{a(\epsilon_1)}^A$ is then the part of $\Sigma$ in between the green dots and described by the green arrow. In the second case, pictured right, the radial distance between $S_{\rho_)}$ and $S_{A}$ is less than or equal to $d_{g_S}(\widetilde{S}_{\rho_0},\widetilde{S}_{A})$ so that $ \ell_{\bg}(\gamma) \leq d_{g_S}(\widetilde{S}_{\rho_0},\widetilde{S}_{A}) $ and $\gamma$ terminates at $\rho_0$. It is this second case in which the control by $\epsilon_1$ is crucial for ensuring the set $U_{a(\epsilon_1)}^A$ does not extend to the boundary. In both cases, the radial depth of $\uep$ is strictly less than $d_{g_S}(\widetilde{S}_{\rho_0},\widetilde{S}_{A})$ in a manner controlled by $\epsilon_1$.}
    \label{fig:set diagram}
\end{figure}

We define a restricted Jang space by
\begin{align}
    \Sigma_{\epsilon_1} := \left\{ (\rho,\theta) \in \Sigma \spa a(\epsilon_1) \leq \rho \leq A \right\}.
\end{align}
The construction of the subsets $\uep$ and $\Sigma_{\epsilon_1}$ succeed in keeping our study of convergence away from the boundary where we do not have specific control over the geometry. We have chosen $\uep$ specifically to control the radial length, which prevents us from falling down a cylindrical region near the boundary, see Figure \ref{fig:set diagram}. However, to justify this restriction, we need to include as part of our argument that $a(\epsilon_1)$ approaches $\rho_0$ as we take $\epsilon_1$ and $\delta$ to zero. We are now ready to state our theorem. 

\begin{theorem}\label{theorem: IF convergence of Jang}
    Suppose $(M,g,k)$ is a spherically symmetric initial data set satisfying the dominant energy condition \eqref{assumption: DEC}, asymptotic flatness given by \eqref{ass: asym flat}, and the uniform bounded outermost expansion condition \eqref{assumption: bounded expansions.}. Further, suppose that $m_{ADM} = (1+\delta)m_0$. Then, there exists a Riemannian manifold $(\Sigma, \bg)$ diffeomorphic to $(M,g)$ with graphical isometric embeddings
    \begin{equation}
    \begin{aligned}
        G: (M,&g) \to (\Sigma \times \IR, \bg - \phi^2dt^2), \quad G(x) = (x, f(x)), \mbox{ and} \\
        &g = G^*(\bg - \phi^2dt^2) = \bg - \phi^2df^2
    \end{aligned}
    \end{equation}
    for $\phi$ as in \eqref{eq: set phi} such that the static spacetime $(\Sigma \times \IR, \bg - \phi^2 dt^2)$ converges to Schwarzschild spacetime 
    \begin{align}
    \left(\sch(m_0) \times \IR, g_S - \phi^2 dt^2\right)
    \end{align}
    in the sense that the restricted base manifold $\Sigma_{\epsilon_1}$ is close to $\sch(m_0)$ for $\epsilon_1$ and $\delta$  very small. In other words, for each arbitrarily small $\epsilon > 0$ there exists $\epsilon_1$ and $\delta$ small enough such that 
    \begin{align}
        d_{\mv\mf}(\uep, \widetilde{U}_{a(\epsilon_1)}^A) < \epsilon
    \end{align}
    where $d_{\mv\mf}(\:, \:)$ is the volume preserving intrinsic flat distance. Moreover, we have the following convergences of the warping factor and the second fundamental form:
    \begin{align}
        \norm{\phi - \phi_S}_{L_{loc}^2(\Sigma_{\epsilon_1},\bg)} \to 0 \mbox{ and } \norm{\sqrt{\phi}(k - \pi)}_{L^2(\Sigma,\bg)} \to 0.
    \end{align}
\end{theorem}

\begin{rem}
We stress that our data $(M,g,k)$ and subsequent Jang surface depends on $\delta$, so that the uniformity of our assumptions is determined with respect to $\delta$. It is possible to prove the theorem for a sequence of inital data $(M_j, g_j, k_j)$; moreover, the theorem holds if the $m_0^j = m_0(M_j)$ is allowed to vary as long as we assume that $m_{ADM}^j = (1+\delta_j)m_0^j \to m_0$ for a sequence of $\delta_j$ converging to zero. In particular, this forces $m_0^j \to m_0$. However, these considerations require more complicated notation and technical details without any fundamental differences of proof, so we omit them from this paper. 
\end{rem}

\section{Background}\label{background}

We review briefly the intrinsic flat convergence as well as the results we will need concerning the Jang equation.

\subsection{Intrinsic Flat Convergence}

The intrinsic flat convergence studies how similar spaces are in terms of their highest order integral properties, in contrast to the Gromov-Hausdorff convergence which operates at the level of geodesics. By definition, the intrinsic flat distance is the smallest flat distance between the integral currents defined by two manifolds under all possible pushforwards to all possible metric spaces. In 2020, Allen, Perales, and Sormani proved an explicit estimate if the volumes are bounded above and the distances bounded below \cite{Allen2020}. Also in 2020, Allen and Perales adapted this argument to the case of Riemannian manifolds with continuous metrics and boundary \cite{Allen2020a}. As our sets $\uep$ and $\wuep$ are compact with boundary, this is the theorem we will apply. 

\begin{theorem}[Volume Above, Distance Below with Boundary]\label{vadb bdy}
    Let $U$ be an oriented and compact manifold, $U_1 = (U,g_1)$ and $U_0 = (U,g_0)$ be continuous Riemannian manifolds with $\mbox{Diam}(U_1) \leq D$, $\vol(U_1) \leq V$, $\vol(\pd U_1) \leq A$ and $F: U_1 \to U_0$ a biLipschitz and distance non-increasing map with a $C^1$ inverse. Let $W_1 \subset U_1$ be a measurable set with 
    \begin{align}
        \vol(U_1 \setminus W_1) \leq V_1
    \end{align}
    and assume there exists an $\alpha_1 > 0$ so that for all $x,y \in W_j$,
    \begin{align}\label{est: dist}
        d_1(x,y) \leq d_0(F_1(x),F_1(y)) + 2\alpha_1
    \end{align}
    and that $h_1 \geq \sqrt{2\alpha_1 D + \alpha_1^2}$. Then
    \begin{align}
        d_{\mf}(U_0,U_1) \leq 2V_1 + h_1 V + h_1 A.
    \end{align}
\end{theorem}
The theorem is proven by constructing an explicit metric space $Z$ in which $U_0$, $U_1$ embed in a distance preserving manner. Although we will find that the diffeomorphism $F_1: \Sigma \to \sch(m_0)$ which restricts to a map $F_1: \uep \to \wuep$ is not necessarily distance decreasing, spherical symmetry allows us to construct an intermediary metric space for which there exists a distance non-increasing diffeomorphism. To show this, we will reference the following lemma from \cite{Allen2020}.
\begin{lemma}\label{lem: met dist}
    Let $U_1=(U,g_1)$ and $U_0 = (U,g_0)$be Riemannian manifolds and $F: U_1 \to U_0$ be a $C^1$ diffeomorphism. Then
    \begin{align}
        g_0(dF(v),dF(v)) \leq g_1(v,v) \mbox{ for all } v \in TU_1
    \end{align}
    if and only if
    \begin{align}
        d_0(F(p),F(q)) \leq d_1(p,q) \mbox{ for all } p,q \in U_1.
    \end{align}
\end{lemma}
We note that the proof of this lemma may be extended to continuous metrics. For a more thorough overview of the intrinsic flat distance, we refer the reader to Section 2.3 of \cite{Allen2020}.

The volume preserving instrinsic flat distance adds a term that controls for the convergence of the global volumes:
\begin{align}
    d_{\mv\mf}(\Omega, \widetilde{U}_{a(\epsilon_1)}^A) = d_{\mf}(\Omega, \widetilde{U}_{a(\epsilon_1)}^A) + |\mbox{Vol}_{g}(\Omega) - \mbox{Vol}_{g'}(\widetilde{U}_{a(\epsilon_1)}^A)|.
\end{align}
Portegies and Jauregue-Lee have shown that volume preserving intrinsic flat convergence implies that the volumes of balls centered on a convergent sequence converge to the volume of the ball in the limit space centered at the limit point \cite{Jauregui2019,Portegies2015}. It is stronger than the unmodified intrinsic flat convergence, but in our case the work needed to upgrade intrinsic flat convergence to volume preserving intrinsic flat convergence is minimal.

\subsection{The Penrose Inequality via the Generalized Jang Equation}

Our approach follows closely that of Bray and Khuri in \cite{Bray2010}. In this section, we restate with discussion the theorems from their generalized Jang approach to the Penrose inequality. As we use multiple elements from their proof of the Penrose inequality and as their result is not stated or proven for general dimension, we reproduce their proof of the Penrose inequality and generalize it to higher dimension. 

First, we state the scalar curvature of the generalized Jang equation. The reader familiar with the expression for the scalar curvature for the trivial extension of the Jang equation should note the differences appear primarily in the divergence term, which is multiplied by the inverse of the warping factor. Indeed, we have that the scalar curvature of the Jang surface, denoted by $\bar{R}$ is given by
	\begin{align}\label{eq: scalar curv}
		\bar{R} = 16\pi(\mu-J(w)) + |h-K|_{\Sigma}|_{\bg}^2 + 2|q|_{\bg}^2 - 2 \phi^{-1}\mbox{div}_{\bg}(\phi q),
	\end{align}
	where $h$ is the second fundamental form, $K|_\Sigma$ is the restriction to $\Sigma$ of the extended tensor $K$, $q$ is a 1-form, and $w$ is a vector with $|w|_g \leq 1$ given by
	\begin{align}
		w = \frac{f^i\pd_{x^i}}{\sqrt{\phi^{-2} + |\nabla_g f|^2 }}, \mbox{ and }
		q_i = \frac{g^{ij}f_{,i} }{\sqrt{\phi^{-2} + |\nabla_g f|^2 }}(h_{ij} - (K|_\Sigma)_{ij}).
	\end{align}
By \eqref{eq: scalar curv}, we can see that once we assume the dominant energy condition \eqref{assumption: DEC}, our the generalized Jang scalar curvature is positive modulo the divergence term. The crux of the stability result we prove here rests on studying the behavior of this divergence term. Although we do not get pointwise control, $L^2$ control is enough to prove the intrinsic flat convergence. 

Another difference between \eqref{eq: scalar curv} and the expression for the trivial Jang extension appears in the second term:
\begin{align}
    |h-K|_{\Sigma}|_{\bg}^2.
\end{align}
Note that $h$ is the second fundamental form of $\Sigma$ as the image of $M$ in $(M \times \IR, g+\phi^2 dt^2)$, not of $M$ inside $(\Sigma \times \IR, \bg- \phi^2 dt^2)$ which we have denoted by $\pi$. Moreover, $K|_{\Sigma}$ is the extension of $k$, not $k$ itself. In order to show that $\pi - k$ is somehow close to zero, we need that
\begin{align}\label{eq: 2nd fund forms}
    h-K|_{\Sigma} = \pi - k.
\end{align}
This is the content of the next lemma, which otherwise appears in Appendix B of \cite{Bray2010}. We leave out the details of computations appearing in Appendix A of \cite{Bray2010}.

\begin{lemma}\label{lem: 2nd fund forms}
Let $h$ be the second fundamental form of the Jang surface $(\Sigma,\bg)$ inside of $(M\times \IR, g+ \phi^2 dt^2)$ and let $\pi$ be the second fundamental form of $G(M)$ inside of $(\Sigma \times \IR, \bg - \phi^2 dt^2)$. Recall that $K$ is the extension of the initial data $k$ given by \eqref{eq: extension for K}. Then we have \eqref{eq: 2nd fund forms}.
\end{lemma}
\begin{proof}
    Let 
    \begin{align}
        X_i = \pd_{x^i} + f_{,i}\pd_{x^{n+1}}, \quad i = 1, 2, 3
    \end{align}
    be tangent vectors to $\Sigma$. Then by equation (42) of \cite{Bray2010}, 
    \begin{align}
        h_{ij} = \frac{\nabla_{ij}f + (\log\phi)_{,i} f_{,j} + (\log \phi)_{,j} f_{,i} + g^{\ell p} \phi \phi_{,\ell} f_{,p} f_{,i} f_{,j} }{\sqrt{\phi^{-2}  + \left| \nabla_g f \right|^2 }}
    \end{align}
    where $\nabla_{ij}$ represents covariant differentiation with respect to $g$. Meanwhile, if $\bar{\nabla}$ represents covariant differentiation with respect to $\bg$, then
    \begin{align}
         \pi_{ij} := \frac{\bar{\nabla}_{ij}f + (\log\phi)_{,i} f_{,j} + (\log \phi)_{,j} f_{,i} - \bg^{\ell p} \phi \phi_{,\ell} f_{,p} f_{,i} f_{,j} }{\sqrt{\phi^{-2}  + \left| \nabla_{\bg} f \right|^2 }}.
    \end{align}
    Calculations making use of the equality $g = \bg - \phi^2 df^2$ show that 
    \begin{align}
    \begin{split}
        \nabla_{ij}f + (\log\phi)_{,i} f_{,j} &+ (\log \phi)_{,j} f_{,i} + g^{\ell p} \phi \phi_{,\ell} f_{,p} f_{,i} f_{,j}
        \\
        &=  \frac{\phi^{-2}}{\phi^{-2} - \left| \nabla_{\bg} f \right|^2}\left(\bar{\nabla}_{ij}f + (\log\phi)_{,i} f_{,j} + (\log \phi)_{,j} f_{,i}  \right) 
        \end{split}
    \end{align}
    and
    \begin{align}
        \frac{1 }{\sqrt{\phi^{-2}  + \left| \nabla_g f \right|^2 }} = \phi^2\sqrt{\phi^{-2} - \left| \nabla_{\bg} f \right|^2 }.
    \end{align}
    It follows that
    \begin{align}
    \begin{split}
        h_{ij} &= \frac{\sqrt{\phi^{-2} - \left| \nabla_{\bg} f \right|^2 }}{\phi^{-2} - \left| \nabla_{\bg} f \right|^2}\left(\bar{\nabla}_{ij}f + (\log\phi)_{,i} f_{,j} + (\log \phi)_{,j} f_{,i} \right) \\
        &= \pi_{ij} + \frac{\bg^{\ell p} \phi \phi_{,\ell} f_{,p} f_{,i} f_{,j} }{\sqrt{\phi^{-2}  + \left| \nabla_{\bg} f \right|^2 }} \\
        &= \pi_{ij} + \frac{ \bracket{\phi \nabla_g \phi, \nabla_g f}_g  }{\sqrt{\phi^{-2}  + \left| \nabla_{g} f \right|^2 }} f_{,i} f_{,j}.
    \end{split}
    \end{align}
    
    On the other hand, we have
    \begin{align}
        (K|_{\Sigma})_{ij} = K(X_i,X_j) = k_{ij} + \frac{ \bracket{\phi \nabla_g \phi, \nabla_g f}_g  }{\sqrt{\phi^{-2}  + \left| \nabla_{g} f \right|^2 }} f_{,i} f_{,j}.
    \end{align}
    It is immediate then that \eqref{eq: 2nd fund forms} holds.
\end{proof}
Note that \eqref{eq: scalar curv} and Lemma \ref{lem: 2nd fund forms} are both independent of the choice of warping factor $\phi$ and the assumption that the initial data is spherically symmetric. However, the existence and uniqueness of the Jang equations relies heavily on the choice of $\phi$ as \eqref{eq: set phi} and spherical symmetry. Recall that the solution $v$ of the Jang equation which determines the embedding $F$ is given by \eqref{eq: eq for v}. If we assume that the initial data are smooth, satisfy the outermost apparent horizon condition \eqref{assumption: apparent horizon}, and the asymptotics \eqref{ass: asym flat} then there exists a unique solution to \eqref{eq: Jang Equation} given by $v \in C^{\infty}((0,\infty)) \cap C^1([0,\infty))$
so that
\begin{align}\label{eq: bd on v}
    0 \leq |v| \leq 1
\end{align}
which satisfies asymptotics to ensure that the Jang surface $\Sigma$ is also asymptotically flat:
\begin{align}\label{eq: asymp for v}
    |v(r)| + r|v_{,r}(r)| \leq C r^{1-n} \mbox{ as } r \to \infty
\end{align}
for a constant $C$ depending only on $|g|_{C^1((0,\infty))}$ and $|k|_{C^0((0,\infty)) }$. We choose boundary data $\alpha = \pm 1$ so that the Jang solution $v$ will blow up to an apparent horizon and the Jang solution will have the same mass information as the initial data. Now we are ready to state and prove the spacetime Penrose inequality in spherical symmetry. 

\begin{theorem}[Spacetime Penrose Inequality in Spherical Symmetry]\label{Penrose Inequality}
    Suppose $n$-dimensional spherically symmetric initial data $(M,g,k)$ satisfies \eqref{assumption: apparent horizon}, and \eqref{ass: asym flat}. Suppose further that the initial data satisfies the dominant energy condition \eqref{assumption: DEC}. Then 
    \begin{align}
        m_{ADM} \geq m_0
    \end{align}
    and if $m_{ADM} = m_0$ then $(M,g)$ arises as an isometric embedding into Schwarzschild spacetime with second fundamental form $k$.
\end{theorem}
\begin{proof}
We can solve the Jang equation for $(M,g,k)$ to get the Jang surface $(\Sigma, \bg, h)$ as in \eqref{jang as image under f}. Moreover, the asymptotics given by \eqref{ass: asym flat} and the choice of horizon boundary for the Jang solution means that $m_{ADM}=m_{ADM}(\Sigma)$ and $m_0 = m_0(\Sigma)$ so that it suffices to prove the inequality for $(\Sigma, \bg)$.

For the first part of the proof, we work in arc length radial coordinates denoted by $(s,\theta)$ on $(\Sigma, \bg)$. In spherical symmetry, the Hawking mass (or Misner-Sharp mass) of the $(n-1)$-sphere $S_s$ generalizes as
\begin{align}
    m(s) &:= \frac{1}{2}\rho^{n-2}(1-\rho_{,s}^2) \label{eq: mass line 1}\\
    &= \frac{1}{2}\left(\frac{|S_s|}{\omega_{n-1}} \right)^{\frac{n-2}{n-1}}\left[1 - \frac{1}{(n-1)^2(\omega_{n-1}^2 |S_s|^{n-3})^{\frac{1}{n-2}}}\int_{S_s} H_{S_s}^2 \right] \label{eq: mass line 2}
\end{align}
where $H_{S_s}$ in $\bg$ as computed in \eqref{eq: mean curv sphere}. We also note that in spherical symmetry the scalar curvature of $\bg$ is
\begin{align}\label{eq: scalar curv 2}
    \bar{R} &= \frac{n-1}{\rho^2}\left((n-2)(1-\rho_{,s}^2) -2\rho\rho_{,ss}  \right).
\end{align}

Differentiate the expression for $m(s)$ in \eqref{eq: mass line 1} with respect to $s$ and use the fundamental theorem of calculus to get that
\begin{align}\label{the one}
    m(\infty) - m(0) &=  \int_0^\infty \frac{1}{2(n-1)}\rho^{n-1}\rho_{,s}\bar{R} ds 
\end{align}
where we have simplified the expression using \eqref{eq: scalar curv 2}. We leverage spherical symmetry to write the integral over $\Sigma$ with volume form $\d\omega_{\bg}$ and then substitute the formula for $\bar{R}$ given by $\eqref{eq: scalar curv}$ to get
\begin{align}
    m_{ADM} - m_0 &= \frac{1}{2\omega_{n-1}(n-1)}\int_\Sigma \rho_{,s} \bar{R} d \omega_{\bg} \label{eq: mass term}  \\
    &= \frac{1}{2\omega_{n-1}(n-1)}\int_\Sigma  \phi \left(16\pi(\mu-J(w)) + |h-K|_{\Sigma}|_{\bg}^2 + 2|q|_{\bg}^2 \right)d \omega_{\bg} \label{eq: pos term} \\
    & \quad -\frac{1}{\omega_{n-1}(n-1)}\int_\Sigma  \mbox{div}_{\bg}(\phi q)  d \omega_{\bg}. \label{eq: div term}
\end{align}
Assumption \eqref{assumption: apparent horizon} guarantees that the mean curvature of $S_s$ is positive, which by \eqref{eq: set phi} implies that $\phi$ is positive as well. This combined with the dominant energy condition \eqref{assumption: DEC} gives that \eqref{eq: pos term} is strictly positive. We can apply Stokes theorem to \eqref{eq: div term} and arrive at
\begin{align}
    m_{ADM} - m_0 \geq  -\frac{1}{\omega_{n-1}(n-1)}\int_{S_{\infty} \cup S_0} \phi \bg( q, n_{\bg})  d \sigma_{\bg}
\end{align}
where $n_{\bg}$ is the outer unit normal and $d\sigma_{\bg}$ is the appropriate surface form. 

The computation in Appendix C of \cite{Bray2010} shows that the part of $q$ normal to spheres can be expressed in $r$ coordinates as
\begin{align}
    q(\pd_r) = -2\sqrt{g_{11}}\frac{v}{1-v^2} \left( \frac{\rho_{,r}}{\rho\sqrt{g_{11}}} v - k(\eta,\eta) \right),
\end{align}
and that in the diagonal metric, $q_i = 0$ for $i > 1$. For a sphere of radius $r$,
\begin{align}
\begin{split}
   -\frac{1}{\omega_{n-1}(n-1)}\int_{S_r} \phi \bg( q, n_{\bg})  d \sigma_{\bg}
   &=  \frac{\rho^{n-1}\phi q_1}{(n-1)\sqrt{g_{11} + \phi^2f_{,r}^2 }}   \\
   &= \pm \frac{2\rho_{,r}v}{\sqrt{g_{11}}}\left( \frac{\rho_{,r}}{\rho\sqrt{g_{11}}} v - k(\eta,\eta) \right) 
\end{split}
\end{align}
depending on whether the boundary is a past or future horizon. By the preservation of the apparent horizon under the choice of $v$,
\begin{align}
    \frac{\rho_{,r}}{\rho\sqrt{g_{11}}} v - k(\eta,\eta) = 0 
\end{align}
on the inner boundary so that $Q$ vanishes on $S_0$. Moreover, the asymptotics for $v$ given by \eqref{eq: asymp for v} and the fall off conditions in \eqref{ass: asym flat} guarantee that $Q(\infty) = 0$. It follows that
\begin{align}
    m_{ADM} - m_0 \geq 0
\end{align}
which proves the inequality.

In the case of equality, the evaluation of the boundary term still holds so that from \eqref{eq: pos term} we have
\begin{align}
    0 = \frac{1}{2\omega_{n-1}(n-1)}\int_\Sigma  \phi \left(16\pi(\mu-J(w)) + |h-K|_{\Sigma}|_{\bg}^2 + 2|q|_{\bg}^2 \right)d \omega_{\bg}
\end{align}
which immediately implies that 
\begin{align}
    \mu - |J|_g \equiv 0, \quad h-K|_{\Sigma} \equiv 0, \mbox{ and } \quad q \equiv 0.
\end{align}
It follows that $\bar{R} = 0$, and from the time symmetric Penrose inequality we deduce that $\bg \cong g_S$ so that, in particular, $\phi = \phi_S$. Moreover,
\begin{align}
    g = \bg - \phi^2 df^2 = g_S - \phi^2 df^2
\end{align}
so that the graph map $G: M \to \sch(M_0) \times \IR$ is an isometric embedding of the initial data into Schwarschild spacetime. 

Finally, by Lemma \ref{lem: 2nd fund forms}, $h-K|_{\Sigma} \equiv 0$ implies that $\pi - k \equiv 0$, and $k$ is indeed the second fundamental form of the embedding into Schwarzschild spacetime.
\end{proof}

\section{Preliminaries}\label{prelims}

We now have all the background we need to set up the proof of Theorem \ref{theorem: IF convergence of Jang}. Recall that $U_{\rho_0}^A = \{ (\rho,\theta) \in \Sigma \spa \rho_0 \leq \rho \leq A \}$ and $\uep =\{ (\rho,\theta) \in \Sigma \spa a(\epsilon_1) \leq \rho \leq A \}$ where 
\begin{align}
    a(\epsilon_1) := \min\{\rho \spa (\rho,\theta) \in T_{\widetilde{D}}(S_A) \} + \epsilon_1
\end{align}
for 
\begin{align}
    \widetilde{D} = \min\left\{d_{g_S}(x,y) \spa x \in \widetilde{S}_A, y \in \pd\sch(m_0) = \widetilde{S}_{\rho_0} \right \}.
\end{align}
Recall also that we may write our Jang metric as 
\begin{align}
    \bg = \frac{1}{\rho_{,s}^2} d\rho^2 + \rho^2 d\Omega^2.
\end{align}

We begin with some preliminary results to control the warping factor $\phi$ and the divergence term in the scalar curvature equation \eqref{eq: scalar curv}. The control over the divergence term is essential for the remainder of our argument. As a consequence, we prove the local convergence of $\phi$ to $\phi_S$. Our first result gives a local upper bound for $\phi$, for which we need the bounded expansion condition.

\begin{lemma}\label{lem: bound on phi}
If we assume the initial data $(M,g,k)$ are smooth, asymptotically flat, spherically symmetric, and satisfy the apparent horizon condition, we can solve the Jang equation for $\phi = \rho_{,s}$ where $s$ is the arc length parameter on $(\Sigma,\bg)$. Additionally, if we assume the bounded expansion condition \eqref{assumption: bounded expansions.}, we have that 
\begin{align}\label{eq: bound on phi}
    \phi \leq AC
\end{align}
on $U_{\rho_0}^A$.
\end{lemma}
\begin{proof}
We can re-write $\phi$ in $r$ coordinates by noting that
\begin{align}
    \frac{dr}{ds} &= \frac{1}{\sqrt{g_{11} + \phi^2 f_{,r}^2}} = \frac{\sqrt{1-v^2}}{\sqrt{g_{11}} }
\end{align}
so that, because $\rho_{,r}$ is positive, 
\begin{align} 
    \phi = \frac{\sqrt{1-v^2}}{\sqrt{g_{11}} } \rho_{,r} \leq \frac{\rho_{,r}}{\sqrt{g_{11}}}.
\end{align}
As described previously, we can deduce from the the bounded expansion condition \eqref{assumption: bounded expansions.} that
\begin{align}
    H_{S_r}^{g} = \frac{2 \rho_{,r}}{\rho \sqrt{g_{11}}} \leq C
\end{align}
and thus at any given $r \in \Sigma$, $\frac{\rho_{,r}}{\sqrt{g_{11}}} \leq \rho C $ so that on $U_{\rho_0}^A$,
\begin{align}
    \phi \leq  \frac{\rho_{,r}}{\sqrt{g_{11}}} \leq AC.
\end{align}
\end{proof}

Now we seek to control the divergence term in the expression for scalar curvature. Following \eqref{the one}, we have in radial coordinates that
\begin{align}
    \begin{split}
        m(\rho) - m_0 &= \frac{1}{2\omega_{n-1}(n-1)}\int_{U_{\rho_0}^\rho} \rho_{,s} \bar{R} d \omega_{\bg} \\
        &= \frac{1}{2\omega_{n-1}(n-1)}\int_{U_{\rho_0}^\rho}  \phi \left(16\pi(\mu-J(w)) + |h-K|_{\Sigma}|_{\bg}^2 + 2|q|_{\bg}^2 \right)d \omega_{\bg} \\
        &\quad -\frac{1}{\omega_{n-1}(n-1)}\int_{S_{\rho}} \phi \bg( q, n_{\bg})  d \sigma_{\bg} \\
        &= \frac{1}{2\omega_{n-1}(n-1)}\int_{U_{\rho_0}^\rho}  \phi \left(16\pi(\mu-J(w)) + |h-K|_{\Sigma}|_{\bg}^2 + 2|q|_{\bg}^2 \right)d \omega_{\bg} \\
        &\quad -\frac{\rho^{n-1}}{n-1} \phi^2 q_1
    \end{split}
\end{align}
where
\begin{align}
    q_1 := q(\pd_{\rho}).
\end{align}

We may therefore write
\begin{align}\label{eq: hawk mass decom}
    m(\rho) &= P(\rho) + Q(\rho) + m_0 
\end{align}
where
\begin{align}\label{eq: restricted pos term}
    P(\rho) = \frac{1}{2\omega_{n-1}(n-1)}\int_{U_{\rho_0}^\rho} 16\pi(\mu-J(w)) + |h-K|_{\Sigma}|_{\bg}^2 + 2|q|_{\bg}^2 d \omega_{\bg} 
\end{align}
and 
\begin{align}
\begin{split}
    Q(\rho) &=-\frac{\rho^{n-1}}{n-1} \phi^2 q_1.
\end{split}
\end{align}
We can see by this definition that $P$ is a positive, increasing function of $\rho$. The computations in Theorem \ref{Penrose Inequality} tell us that
\begin{align*}
    \delta m_0 = m_{ADM} - m_0 = \lim_{\rho \to \infty} P(\rho)
\end{align*}
so that $0 \leq P(\rho) \leq \delta m_0$. Our next lemma gives control over $Q$.

\begin{lemma}\label{lem: analysis of Q}
If the initial data $(M,g,k)$ are smooth, asymptotically flat as in \eqref{ass: asym flat}, and spherically symmetric and satisfy the dominant energy condition \eqref{assumption: DEC} and bounded expansion condition \eqref{assumption: bounded expansions.}, we have
\begin{align}\label{eq: bound on q}
    \norm{Q}_{L^2(U_{\rho_0}^A,\bg)}^2 \leq \delta \omega_{n-1} m_0 A^{2n-1}C
\end{align}
and, in particular, $\norm{Q}_{L^2(U_{\rho_0}^A,\bg)} \to 0$ as $\delta \to 0$.  
\end{lemma}
\begin{proof}
By the analysis in Theorem \ref{Penrose Inequality},
\begin{align}
\begin{split}
    \delta m_0 &= m_{ADM} - m_0 \\
    &\geq \frac{1}{2\omega_{n-1}(n-1)}\int_\Sigma  \phi \left(16\pi(\mu-J(w)) + |h-K|_{\Sigma}|_{\bg}^2 + 2|q|_{\bg}^2 \right)d \omega_{\bg} \\
    &\geq  \frac{1}{\omega_{n-1}(n-1)}\int_{U_{\rho_0}^A}  \phi |q|_{\bg}^2 d \omega_{\bg} \\
    & = \frac{1}{\omega_{n-1}(n-1)}\int_{U_{\rho_0}^A}  \phi^3 q_1^2 d \omega_{\bg}. 
\end{split}
\end{align}
From this inequality, we get that
\begin{align}
\begin{split}
    \delta \omega_{n-1} m_0 (n-1) &\geq \int_{U_{\rho_0}^A} \left(\frac{\rho^{(n-1)}}{n-1}\phi^2 q_1\right)^2 \frac{(n-1)^2}{\rho^{2(n-1)}\phi}  d \omega_{\bg}  \\
    &\geq \frac{(n-1)^2}{A^{2n-1}C} \int_{U_{\rho_0}^A} |Q|^2   d \omega_{\bg}
\end{split}
\end{align}
where in the last line we have used \eqref{eq: bound on phi}. It follows that
\begin{align}
    \norm{Q}_{L^2(U_{\rho_0}^A,\bg)}^2 \leq \frac{\delta \omega_{n-1} m_0 A^{2n-1}C}{n-1}
\end{align}
and we have the result.
\end{proof}

The motivation for controlling the divergence term is made more clear if we re-write the metric $\bg$ in the form of a Schwarzschild metric. Indeed, if for some function $m(\rho)$ we set
\begin{align}\label{eq: met expres}
    \bg = \frac{1}{\phi^2} d\rho^2 + \rho^2 d\Omega^2=  \left( 1 - \frac{2m(\rho)}{\rho^{n-2}} \right)^{-1}d\rho^2 + \rho^2 d\Omega^2 
\end{align}
we solve to get that $m(\rho) = \frac{1}{2} \rho^{n-2}(1-\rho_{,s}^2)$, which is exactly the Hawking mass at $S_{\rho}$ as in \eqref{eq: mass line 1}. From this, and the decomposition of the Hawking mass as in \eqref{eq: hawk mass decom}, we see that
\begin{align}\label{eq: met expres 2}
    \bg = \left( 1 - \frac{2(m_0 + P(\rho) + Q(\rho) )}{\rho^{n-1}}  \right)^{-1}d\rho^2 + \rho^2 d\Omega^2 
\end{align}
and the difference between $\bg$ and $g_S$ is then controlled by the terms $P(\rho)$ and $Q(\rho)$, which determine the deviation of the Hawking mass from the half area radius $m_0$. We have seen that $0 \leq P(\rho) \leq \delta m_0$; moreover, $P$ is increasing. Therefore, the behavior of the metric is controlled by $Q$, and the $L^2$ control of $Q$ by $\delta$ is thus essential to proving intrinsic flat convergence.

\section{Convergence of the warping factor and its consequences}\label{warp fact convergence}
Now we turn to the study of the warping factor $\phi$ and show convergence to $\phi_S$ as a consequence of Lemma \ref{lem: analysis of Q}. To prove this result we will need to restrict away from the boundary. Recall that
\begin{align}
    \phi = \sqrt{1 - \frac{2m(\rho)}{\rho^{n-2}}} \quad \mbox{ and } \quad \phi_S = \sqrt{1 - \frac{2m_0}{\rho^{n-2}}}.
\end{align}
\begin{lemma}\label{lem: conv of phi}
If the initial data $(M,g,k)$ are smooth, asymptotically flat, and spherically symmetric and satisfy the dominant energy condition \eqref{assumption: DEC} and bounded expansion condition \eqref{assumption: bounded expansions.}, then for any $\epsilon_1 > 0$ and $\epsilon > 0$ there is $\delta(\epsilon_1,\epsilon) > 0$ small enough so that $\norm{\phi - \phi_S}_{L^2(\uep,\bg)} < \epsilon$. 
\end{lemma}
\begin{proof}
    We compute
	\begin{align}
		\int_{\uep} \left| \phi - \phi_S \right|^2 d\omega_{\bg} & = \int_{\uep} \frac{4}{\rho^{2(n-2)}(\phi + \phi_S)^2} \left| m_0 - m(\rho) \right|^2 d\omega_{\bg}
    \end{align}
    by rationalizing $(\phi - \phi_S)^2$. Now, because $\rho \geq a(\epsilon_1)$, $\phi \geq 0$, and $\phi_S \geq \phi_S(a(\epsilon_1))$ on $\uep$, we may estimate that on $\uep$,
    \begin{align}
    \begin{split}
        \frac{4}{\rho^{2(n-2)}(\phi + \phi_S)^2} &\leq \frac{4}{(a(\epsilon_1)^{n-2} \phi_S(a(\epsilon_1)))^2} \\
        &= \frac{4}{(a(\epsilon_1))^{n-2}(a(\epsilon_1)^{n-2}-2m_0)}. 
    \end{split}
    \end{align}
    Using this and the decomposition \eqref{eq: hawk mass decom}, we have that 
    \begin{align}
    \begin{split}
		\int_{\uep} \left| \phi - \phi_S \right|^2 d\omega_{\bg} &\leq \frac{4}{(a(\epsilon_1))^{n-2}(a(\epsilon_1)^{n-2}-2m_0)}  \int_{\uep}|Q + P |^2d\omega_{\bg} \\
        &\leq\frac{8}{(a(\epsilon_1))^{n-2}(a(\epsilon_1)^{n-2}-2m_0)}  \left(\int_{\uep}|Q|^2 d\omega_{\bg} + \int_{\uep}|P |^2d\omega_{\bg} \right).
    \end{split}
	\end{align}
	We know that $0 \leq P(\rho) \leq \delta$ for any $\rho \in \Sigma$ and we can apply the estimate \eqref{eq: bound on q} from Lemma \ref{lem: analysis of Q} to see that
	\begin{align}
 \begin{split}
	    \int_{\uep} \left| \phi - \phi_S \right|^2 d\omega_{\bg} & \leq \frac{8}{(a(\epsilon_1))^{n-2}(a(\epsilon_1)^{n-2}-2m_0)}  \left(\delta \omega_{n-1}m_0 A^{2n-1} C + \int_{\uep}\delta^2d\omega_{\bg}\right) \\
        &\leq \frac{8}{(a(\epsilon_1))^{n-2}(a(\epsilon_1)^{n-2}-2m_0)}  \left(\delta \omega_{n-1}m_0 A^{2n-1} C + \delta^2\omega_{n-1}A^{n-1}\widetilde{D}\right)
        \end{split}
	\end{align}
    where we have estimated $\vol(\uep)$ by $\omega_{n-1}A^{n-1}\widetilde{D}$ using the coarea formula. We find therefore that for $\delta$ small enough with respect to $\epsilon_1$ we have $\norm{\phi - \phi_S}_{L^2(\uep,\bg)} < \epsilon$.
\end{proof}

The next result describes the behavior of $a(\epsilon_1)$ as we take $\delta \to 0$. Specifically, we find that $a(\epsilon_1) \to \rho_0 + \epsilon_1$, meaning that $T_{\widetilde{D}}(S_A) \cap U_{\rho_0}^A \subset \Sigma$ will have inner boundary close to the Schwarzschild radius for small enough $\delta$ and $\epsilon_1$. 

\begin{figure}[!ht]
\centering
\begin{tikzpicture}

\draw[->] (0,-1) .. controls (0.5,2) and (3,2) .. (5,2.1);
\node[circle, inner sep=0pt, minimum size = 5pt, fill=blue, label={0: $\rho = \rho_0$}] at (0,-1) {};

\draw[red, thick](0.25,-0.07437)  .. controls (0.85,1.65) and (2.75,2.05) .. (4.5,2.07574) ;

\draw[black, thick, <-] (0.005,-0.07437) -- (0.25,-0.07437);

\node[circle, inner sep=0pt, minimum size = 5pt, fill=green, label={90: $\rho = A$}] at (4.5,2.07574) {};
\node[circle, inner sep=0pt, minimum size = 1pt, fill=red, label={0: $\rho = a(\epsilon_1)-\epsilon_1$}] at (0.25,-0.07437) {};
\node[circle, inner sep=0pt, minimum size = 5pt, fill=green, label={180: $\rho = a(\epsilon_1)$}] at (0.5,0.432282) {};

\draw[->] (0,-1.5) -- (5,-1.5);
\node[circle, inner sep=0pt, minimum size = 0.5pt, fill=black, label={270: $\rho$}] at (4,-1.5) {};

\end{tikzpicture}
\begin{tikzpicture}

\draw[->] (0,-0.07437) .. controls (0.5,2) and (3,2) .. (5,2.1);
\draw (0,-1) -- (0,-0.07437);
\node[circle, inner sep=0pt, minimum size = 5pt, fill=blue, label={0: $\rho = \rho_0$}] at (0,-1) {};

\draw[red, thick](0,-0.07437)   .. controls (0.45,1.65) and (2.25,2.05) .. (4.5,2.07574) ;

\node[circle, inner sep=0pt, minimum size = 5pt, fill=green, label={90: $\rho = A$}] at (4.5,2.07574) {};
\node[circle, inner sep=0pt, minimum size = 5pt, fill=red, label={0: $\rho = a(\epsilon_1)-\epsilon_1 = \rho_0$}] at (0,-0.07437) {};
\node[circle, inner sep=0pt, minimum size = 5pt, fill=green, label={0: $\rho = a(\epsilon_1)$}] at (0.2,0.432282) {};
\draw[->] (0,-1.5) -- (5,-1.5);
\node[circle, inner sep=0pt, minimum size = 0.5pt, fill=black, label={270: $\rho$}] at (4,-1.5) {};
\end{tikzpicture}
\caption{Here we show the nontrivial case of $a(\epsilon_1)$ convergence, which occurs in the first case described in Figure \ref{fig:set diagram}. The black arrow emanating from $\rho = a(\epsilon_1)-\epsilon_1$ in the left diagram demonstrates the convergence of this value to $\rho_0$. The arrow is drawn horizontally to make it clear that, while this convergence occurs by Lemma \ref{lem: a close to rho}, the set $U_{\rho_0}^{a(\epsilon_1)-\epsilon_1}$ will not necessarily collapse. In the case that it does not collapse, this set converges to a cylindrical end, as shown in the diagram on the right. }
    \label{fig:convergence of uep bdy}
\end{figure}
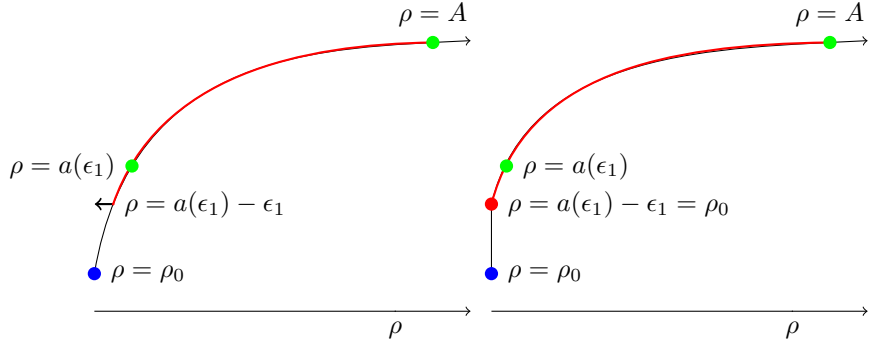

\begin{lemma}\label{lem: a close to rho}
	If the initial data $(M,g,k)$ are smooth, asymptotically flat, and spherically symmetric and satisfy the dominant energy condition \eqref{assumption: DEC} and bounded expansion condition \eqref{assumption: bounded expansions.}, then as $\delta \to 0$ we have that $a(\epsilon_1) \to \rho_0 + \epsilon_1$.
\end{lemma}
\begin{proof}
	Recall that we have defined 
    \begin{align}
        a(\epsilon_1) := \min\{\rho \spa (\rho,\theta) \in T_{\widetilde{D}} \cap \uep \} + \epsilon_1
    \end{align}
    where $\widetilde{D}$ is the Schwarzschild distance between $\widetilde{S}_A$ and $\pd \sch(m_0)$. In the case that the radial depth of $\upa$ is less than or equal to $\widetilde{D}$, we have that $a(\epsilon_1) = \rho_0 +\epsilon_1$. Otherwise, $a(\epsilon_1) > \rho_0 +\epsilon_1$. In that case, set 
    \begin{align}
        \rho_{\delta} = a(\epsilon_1) - \epsilon_1 = \min\{\rho \spa (\rho,\theta) \in T_{\widetilde{D}} \cap U_{\rho_0}^A \}.
    \end{align}
    We assume, for contradiction, that convergence does not occur, so that $\rho_{\delta} - \rho_0 \geq c$ as $\delta \to 0$ for some small $c > 0$. By definition of $U_{\rho_{\delta}}^A = T_{\widetilde{D}} \cap U_{\rho_0}^A$, we have that the radial distance between $S_A $ and $S_{\rho_{\delta}}\subset \Sigma$ is
		\begin{align}
			\int_{\rho_{\delta}}^A \frac{1}{\phi}d\rho  =\widetilde{D}
		\end{align} 
    where $S_A \cup S_{\rho_{\delta}} = \pd \uep$. We use this to estimate the difference between $\widetilde{D}$ and the radial distance between $\widetilde{S}_A$ and $\widetilde{S}_{\rho_{\delta}} \subset \sch(m_0)$ as follows:
	\begin{align}
 \begin{split}
		\left| \widetilde{D} - \int_{\rho_{\delta}}^A\frac{1}{\phi_S}d\rho \right|  &\leq \int_{\rho_{\delta}}^A \left|\frac{1}{\phi} - \frac{1}{\phi_S} \right|d\rho \\
		&= \int_{\rho_{\delta}}^A \frac{1}{\phi \phi_S} \left|\phi_S - \phi \right|d\rho \\
        &\leq \frac{1}{\omega_{n-1}}\int_{\rho_{\delta}}^A\int_{S_\rho} \frac{1}{\phi_S\rho^{n-1}}|\phi_S - \phi|  \frac{\rho^{n-1}}{\phi} d\rho d\sigma_{\bg} \\
		&\leq \frac{1}{\omega_{n-1}\phi_S(\rho_{\delta})\rho_{\delta}^{n-1}} \int_{U_{\rho_\delta}^A}\left|\phi_S - \phi \right|d\omega_{\bg}
  \end{split}
	\end{align}
	where in the last line we have used that $d\omega_{\bg} = \frac{1}{\phi}\rho^{n-1}d\rho d\sigma_{\bg}$ and estimated from above on $U_{\rho_\delta}^A$. By H\"older's inequality and the coarea formula,
    \begin{align}
    \begin{split}
        \left| \widetilde{D} - \int_{\rho_{\delta}}^A\frac{1}{\phi_S}d\rho \right|^2  &\leq \left(\frac{1}{\omega_{n-1}\phi_S(\rho_{\delta})\rho_{\delta}^{n-1}} \right)^2 \vol_{\bg}(U_{\rho_\delta}^A) \int_{U_{\rho_\delta}^A}\left|\phi_S - \phi \right|^2d\omega_{\bg} \\
        & \leq \left(\frac{1}{\omega_{n-1}\phi_S(\rho_{\delta})\rho_{\delta}^{n-1}} \right)^2 \left(A^{n-1}\widetilde{D} \right) \int_{U_{\rho_\delta}^A}\left|\phi_S - \phi \right|^2d\omega_{\bg}.  
    \end{split}    
    \end{align} 
    Our assumption that $\rho_{\delta}$ stays at least $c$ away from $\rho_0$ means that 
    \begin{align}
        \phi_S(\rho_{\delta}) = \sqrt{1 - \frac{\rho_0^{n-2}}{\rho_{\delta}^{n-2}}} \geq \sqrt{\frac{(\rho_0 + c)^{n-2}-\rho_0^{n-2}}{(\rho_0 + c)^{n-2}} } > 0
    \end{align}
    and so by Lemma \ref{lem: conv of phi} the right hand side becomes arbitrarily small as $\delta \to 0$. It follows that
	\begin{align}
	    \int_{\rho_{\delta}}^A \frac{1}{\phi_S} d\rho \to \widetilde{D} = \int_{\rho_0}^A \frac{1}{\phi_S} d\rho
	\end{align}
	as $\delta \to 0$. However, $\int_{\rho_{\delta}}^A \frac{1}{\phi_S} d\rho$ increases as $\rho_{\delta}$ decreases with maximum $\widetilde{D}$ which is only achieved when $\rho_{\delta} = \rho_0$. We then have a contradiction that $\rho_{\delta} - \rho_0 \geq c > 0$ for all $\delta$.
\end{proof}

In the next proposition, we show that $\norm{\widetilde{U}_{\rho_0}^{a(\epsilon_1)}}_{\mf}$ is arbitrarily small when $\epsilon_1$ and $\delta$ are small enough. In other words, the ``remainder" of the Schwarzschild space is small. 

\begin{prop}\label{prop: size of inner}
    Suppose the initial data $(M,g,k)$ are smooth, asymptotically flat, and spherically symmetric and satisfy the dominant energy condition \eqref{assumption: DEC} and bounded expansion condition \eqref{assumption: bounded expansions.}. Then for any $\epsilon > 0$, there exist $\epsilon_1(\epsilon)$, $\delta(\epsilon_1)$ small enough so that $\norm{\widetilde{U}_{\rho_0}^{a(\epsilon_1)}}_{\mf} < \epsilon$.
\end{prop}
\begin{proof}
    We have that 
    \begin{align}
    \begin{split}
        \norm{\widetilde{U}_{\rho_0}^{a(\epsilon_1)}}_{\mf} &\leq \vol_{g_S}\left(\widetilde{U}_{\rho_0}^{a(\epsilon_1)} \right) \\
        & = \int_{\widetilde{U}_{\rho_0}^{a(\epsilon_1)} }  d\omega_{g_S} \\
        &\leq \omega_{n-1}(a(\epsilon_1))^{n-1} \int_{\rho_0}^{a(\epsilon_1)}\frac{1}{\phi_S} d\rho
    \end{split}
    \end{align}
    where the last line follows by the coarea formula. By Lemma \ref{lem: a close to rho}, for $\delta$ small enough we have $a(\epsilon_1) \leq \rho_0 + 2\epsilon_1$ so that
    \begin{align}
    \begin{split}
         \norm{\widetilde{U}_{\rho_0}^{a(\epsilon_1)}}_{\mf} &\leq \omega_{n-1}(\rho_0 + 2\epsilon_1)^{n-1} \int_{\rho_0}^{\rho_0 + 2\epsilon_1} \sqrt{\frac{\rho^{n-2}}{\rho^{n-2}-\rho_0^{n-2}}} d\rho \\
         &\leq \omega_{n-1}(\rho_0 + 2\epsilon_1)^{n-1} \int_{\rho_0}^{\rho_0 + 2\epsilon_1} \sqrt{\frac{\rho}{\rho-\rho_0}}\sqrt{\frac{\rho^{n-3}}{\sum_{i=0}^{n-3}\rho^i\rho_{0}^{n-3-i}}} d\rho \\
         &\leq \omega_{n-1}(\rho_0 + 2\epsilon_1)^{n-1}\sqrt{\frac{(\rho_0+2\epsilon_1)^{n-3}}{(n-3)\rho_0^{n-3}}}\int_{\rho_0}^{\rho_0 + 2\epsilon_1} \sqrt{\frac{\rho}{\rho-\rho_0}} d\rho
    \end{split}
    \end{align}
    where in the last line we have applied the estimate on $[\rho_0,\rho_0+2\epsilon_1]$ that
    \begin{align}
        \sqrt{\frac{\rho^{n-3}}{\sum_{i=0}^{n-3}\rho^i\rho_{0}^{n-3-i}}} \leq \sqrt{\frac{(\rho_0+2\epsilon_1)^{n-3}}{\sum_{i=0}^{n-3}\rho_0^i\rho_{0}^{n-3-i}}} \leq \sqrt{\frac{(\rho_0+2\epsilon_1)^{n-3}}{(n-3)\rho_0^{n-3}}}.
    \end{align}
    The function $\sqrt{\frac{\rho}{\rho-\rho_0}}$ is integrable on $[\rho,\rho+2\epsilon_1]$, from which it follows that for $\epsilon_1$ and $\delta$ small enough we may obtain 
    \begin{align}
        \norm{\widetilde{U}_{\rho_0}^{a(\epsilon_1)}}_{\mf} < \epsilon.
    \end{align}
\end{proof}

\section{Application of VADB with Boundary}\label{vadb app}

In this section, we prove the next main ingredient for the theorem, ie, that
\begin{align}
    d_{\mf}(\uep,\wuep)
\end{align}
can be made as small as we like. However, we cannot do this directly: recall that, as in \eqref{eq: met expres 2}, we can write each metric as
\begin{align}
    \bg = \left(1 - \frac{2(m_0 + P(\rho) + Q(\rho))}{\rho^{n-2}}  \right)^{-1} d\rho^2 + \rho^2 d\Omega^2 && g_S = \left(1 - \frac{2m_0}{\rho^{n-2}}  \right)^{-1} d\rho^2 + \rho^2 d\Omega^2. 
\end{align}
By inspecting the metrics, we can see that the diffeomorphism $F_1: \uep \to \wuep$ is not distance decreasing. Indeed, we can compute that $g_S \geq \bg$ when $P \geq -Q$; as we have no control over the sign of $Q$ this is certainly a possibility. However, we may define an intermediary metric space which will admit a distance decreasing map from both sets.

Define the space $U_0 = (U,g_0)$ with $U$ diffeomorphic to $\uep$ and $\wuep$. Define $g_0$ as follows. Let $V = \{(\rho,\theta) \in U \spa P(\rho) \geq -Q(\rho) \}$. This is a closed set composed of countably many disjoint annular regions: denote these connected components by $\{V_{j} \}$. For each $V_j$, let $\rho_j = \min\{\rho \spa (\rho,\theta) \in V_j \}$. Let $V' = \bigcup_j (V_j - S_{\rho_j})$ so that we have deleted the innermost circle from each component of $V$. Let $g_0$ be a metric on $U$ defined as follows:
\begin{align}
    g_0 := \begin{cases}
        \bg &\mbox{ when } x \in V' \\
        g_S &\mbox{ when } x \in U \setminus V'.
    \end{cases}
\end{align}
Note that by definition, $\bg = g_S$ on $\pd V'$ so that the metric $g_0$ is continuous and, for every annulus on which $g_0 = \bg$ and $g_0 = g_S$, the metric is smooth up to the outer boundary.. Moreover, $\bg \leq g_S$ on $V'$ and $g_S \leq \bg$ on $U \setminus V'$ so by Lemma \ref{lem: met dist}, the diffeomorphisms $J_1: \uep \to U_0$, $J_2: \wuep \to U_0$ for which $(\rho,\theta) \mapsto (\rho,\theta)$ are distance decreasing. 

In this section, we prove the following proposition.
\begin{prop}\label{prop: inter metric space}
 Suppose the initial data $(M,g,k)$ are smooth, asymptotically flat, and spherically symmetric and satisfy the dominant energy condition \eqref{assumption: DEC} and bounded expansion condition \eqref{assumption: bounded expansions.}. Then, for any $\epsilon, \epsilon_1 > 0$ there exists $\delta(\epsilon,\epsilon_1) > 0$ small enough so that
 \begin{align}
 \begin{split}
     d_{\mf}(U_0, \uep) &< \epsilon/3, \\
     d_{\mf}(U_0, \wuep) &< \epsilon/3
    \end{split}
 \end{align}
     when $m_{ADM} = (1+\delta) m_0$.
\end{prop}

First, we prove bounds on diameters, volumes, and volumes of the boundary for each $\uep$, $\wuep$. Then, we choose the set $W$ on which we can get the estimate \eqref{est: dist} and prove the distance estimate. Finally, we prove the proposition by applying Theorem \ref{vadb bdy}.

\subsection{Volumes, Areas, and Diameters}

We now estimate the volumes, boundary areas, and diameters of the sets $\uep$ and $\wuep$.

\begin{lemma} \label{lemm: volumes of regions}
	Suppose the initial data $(M,g,k)$ are smooth, asymptotically flat, and spherically symmetric and satisfy the dominant energy condition \eqref{assumption: DEC} and bounded expansion condition \eqref{assumption: bounded expansions.}. Then, the volumes of the diffeomorphic subregions $\uep$ and $\wuep$ may be estimated by 
	\begin{align}\label{est: bg vol in}
		\mbox{Vol}_{\bg}(\uep), \mbox{Vol}_{g_S}(\wuep) \leq \omega_{n-1}A^{n-1}\widetilde{D}.
	\end{align}
\end{lemma}
\begin{proof}
    Both $\uep$ and $\wuep$ have a radial depth of at most $\widetilde{D}$ and a largest sphere of radius $A$. By the coarea formula, we have
    \begin{align}
    \begin{split}
        \mbox{Vol}_{\bg}(\uep) &= \int_{\uep} d\omega_{\bg} \\
            &\leq \mbox{Area}_{\bg}(S_A) \int_{\rho_0 + \epsilon_1}^A \frac{1}{\phi} d \rho \\
            & \leq \omega_{n-1}A^{n-1}\widetilde{D}.
    \end{split}
    \end{align}
    The same estimate holds for $\wuep$.
\end{proof}

\begin{lemma}\label{lem: bdy areas}
    Suppose the initial data $(M,g,k)$ are smooth, asymptotically flat, and spherically symmetric and satisfy the dominant energy condition \eqref{assumption: DEC} and bounded expansion condition \eqref{assumption: bounded expansions.}. Then, the boundary areas of the diffeomorphic subregions $\uep$ and $\wuep$ may be estimated by 
	\begin{align}\label{est: bdy}
		\mbox{Vol}_{\bg}(\pd\uep), \mbox{Vol}_{g_S}(\pd\wuep) = \omega_{n-1}(A^{n-1}+a(\epsilon_1)^{n-1}).
	\end{align}
\end{lemma}
\begin{proof}
    This follows by spherical symmetry and the radial values on the boundary.
\end{proof}

\begin{lemma}\label{lemm diams}
	Suppose the initial data $(M,g,k)$ are smooth, asymptotically flat, and spherically symmetric and satisfy the dominant energy condition \eqref{assumption: DEC} and bounded expansion condition \eqref{assumption: bounded expansions.}. Let $D_0 \geq A, \widetilde{D}$. Then
	\begin{align}
		\max\{\mbox{diam}_{\bg}(\uep),\mbox{diam}_{g_S}(\widetilde{U}_{a(\epsilon_1)}^A) \} \leq  4 \pi D_0.
	\end{align}
\end{lemma}
\begin{proof}
	The depth of $\uep$ and $\widetilde{U}_{a(\epsilon_1)}^A$ is at most $2\widetilde{D}$, and the largest symmetry sphere satisfies
	\begin{align}
		\mbox{diam}_{g_S}(S_A) = \mbox{diam}_{g_S}(\widetilde{S}_A) = \pi A.
	\end{align}
	By the triangle inequality the diameters of $\uep$ and $\widetilde{U}_{a(\epsilon_1)}^A$ are no larger than $4 \widetilde{D}+ \pi A \leq 4\pi D_0$. 
\end{proof}

\subsection{Metric Estimate}

Although we will eventually compare the metrics $g_S$ to $g_0$ and $\bg$ to $g_0$, because $g_0$ is an amalgam of $g_S$ and $\bg$ we may first compare these metrics. In this section, we define the the ``good" region $W$ on which $\bg$ will be close to $g_S$ and ``bad" region $B := U - W$ on which the metric is not close to $g_S$. Once again, we can write each metric as
\begin{align}
    \bg = \left(1 - \frac{2(m_0 + P(\rho) + Q(\rho))}{\rho^{n-2}}  \right)^{-1} d\rho^2 + \rho^2 d\Omega^2 && g_S = \left(1 - \frac{2m_0}{\rho^{n-2}}  \right)^{-1} d\rho^2 + \rho^2 d\Omega^2 
\end{align}
which suggests that choosing the set on which $\bg$ and $g_S$ are close is equivalent to choosing the set on which $Q$ is forced to be small in some precise way.

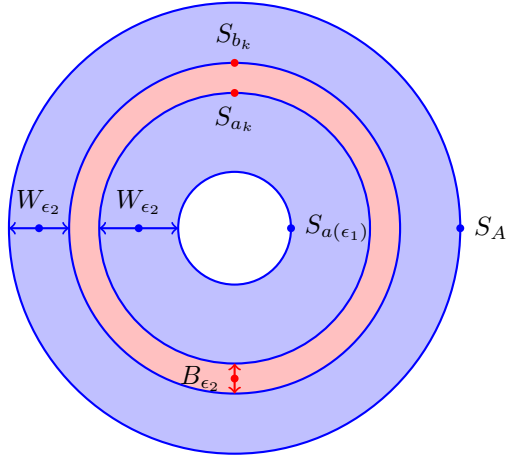
\begin{figure}[h!]
    \centering
    \begin{tikzpicture}
\fill[red,nearly transparent] (1.8,0) 
    arc [radius = 1.8, start angle = 0, delta angle = 360] 
    -- (2.2,0) arc [radius = 2.2, start angle = 360, delta angle = -360] -- cycle;

    \fill[blue,nearly transparent] (0.75,0) 
    arc [radius = 0.75, start angle = 0, delta angle = 360] 
    -- (1.8,0) arc [radius = 1.8, start angle = 360, delta angle = -360] -- cycle;

    \fill[blue,nearly transparent] (2.2,0) 
    arc [radius = 2.2, start angle = 0, delta angle = 360] 
    -- (3,0) arc [radius = 3, start angle = 360, delta angle = -360] -- cycle;


    
    \draw[blue,thick] circle[radius=2.2];
    \draw[blue,thick] circle[radius=3];
    \draw[blue,thick] circle[radius=1.8];
    \draw[blue,thick] circle[radius=0.75];


    \node[circle, inner sep=0pt, minimum size = 3pt, fill=blue, label={0: $S_{a(\epsilon_1)}$}] at (0.75,0) {};

    \node[circle, inner sep=0pt, minimum size = 3pt, fill=red, label={270: $S_{a_{k}}$}] at (0,1.8) {};

    \node[circle, inner sep=0pt, minimum size = 3pt, fill=red, label={90: $S_{b_{k}}$}] at (0,2.2) {};

    \node[circle, inner sep=0pt, minimum size = 3pt, fill=blue, label={0: $S_A$}] at (3,0) {};

    \draw[blue,thick, <->] (-3,0) -- (-2.2,0);
    \draw[blue,thick, <->] (-1.8,0) -- (-0.75,0);
    \draw[red,thick, <->] (0,-2.2) -- (0,-1.8);

    \node[circle, inner sep=0pt, minimum size = 3pt, fill=blue, label={90: $W_{\epsilon_2}$}] at (-2.6,0) {};
    \node[circle, inner sep=0pt, minimum size = 3pt, fill=blue, label={90: $W_{\epsilon_2}$}] at (-1.275,0) {};

    \node[circle, inner sep=0pt, minimum size = 3pt, fill=red, label={180: $B_{\epsilon_2}$}] at (0,-2) {};
\end{tikzpicture}
    \caption{We show a sketch of the regions of interest inside of $\Sigma$. The blue regions labeled $W_{\epsilon_2}$ represent the part of $\uep$ for which $|Q|\leq \epsilon_2$ and the red region is the complement of this set in $\uep$. Note that the spherical symmetry forces the components of $W_{\epsilon_2}$ and $B_{\epsilon_2 }$ to be annular regions as displayed here.}
    \label{fig:good set bad set}
\end{figure}

We introduce a new parameter $\epsilon_2 > 0$ and let 
\begin{align}
\begin{split}
    B_{\epsilon_2} &:= \{ (\rho,\theta) \in \uep \spa |Q(\rho)| > \epsilon_2  \} \mbox{ and }\\
    W_{\epsilon_2} &:= \uep - B_{\epsilon_2}.
    \end{split}
\end{align}
By definition, $B_{\epsilon_2 } - \pd \uep$ is open and composed of perhaps infinitely many disjoint annular components; however, as $\uep$ is second countable, there may be at most countably many of these components. Let $B_k$ be a component of $B_{\epsilon_2 }$ so that $\pd B_k = S_{a_k} \cup S_{b_k}$ with $a_k < b_k$ when $a_k \ne a(\epsilon_1)$ and $b_k \ne A$; otherwise, we might have $a_k = b_k$. We can thus describe $\pd B_{\epsilon_2}$ by the radial values of the inner and outer boundaries of each region -- ie, the set $\{(a_k,b_k) \}$. See Figure \ref{fig:good set bad set}. We denote the diffeomorphic counterparts of these sets in $(\sch(m_0),g_S)$ by $\widetilde{B}_{\epsilon_2}$ and $\widetilde{W}_{\epsilon_2}$. When we identify these regions via the diffeomorphism, we refer to them as $B$ and $W$.

\begin{lemma}\label{lem: metric estimate}
    Suppose the initial data $(M,g,k)$ are smooth, asymptotically flat, and spherically symmetric and satisfy the dominant energy condition \eqref{assumption: DEC} and bounded expansion condition \eqref{assumption: bounded expansions.}. Let $\epsilon_1$, $\epsilon_3 > 0$. Then, for $\epsilon_2(\epsilon_1,\epsilon_3)$ and $\delta(\epsilon_1,\epsilon_2,\epsilon_3)$ small enough, we have 
    \begin{align}\label{eq: metric estimate}
        \bg \leq (1+\epsilon_3)^{2} g_S \mbox{ and } g_S \leq (1+\epsilon_3)^{2} \bg
    \end{align}
    on $W$.
\end{lemma}
\begin{proof}
    The proof is done by direct estimate. First, we find a function $c_1(\rho) \geq 1$ that we can use to get the comparison $\bg \leq c_1(\rho) g_S$. Second, we show that for a given $\epsilon_1$ and $\epsilon_3$ there exist $\delta$ and $\epsilon_2$ small enough so that $c_1(\rho) \leq (1+\epsilon_3)^2$ on $W$. We repeat the procedure to find $c_2(\rho)$ so that $g_S \leq c_2(\rho) \bg$ and $c_2(\rho) \leq (1+\epsilon_3)^2$ on $W$ for $\delta$ and $\epsilon_2$ small enough.

    In $\rho$ coordinates we have
    \begin{align}
    \begin{split}
        \bg &= \left( 1 - \frac{2 m(\rho)}{\rho^{n-2}}  \right)^{-1} d\rho^2 + \rho^2 d\Omega^2  \\
        &= \left( 1 - \frac{2 (m_0 + P(\rho) + Q(\rho) )}{\rho^{n-2}}  \right)^{-1} d\rho^2 + \rho^2 d\Omega^2. 
    \end{split}
    \end{align}
    We have that $P \leq \delta m_0$ on all of $\Sigma$. Restricting to $W$, we get $Q \leq \epsilon_2$ as well. We therefore have the estimate 
    \begin{align}
        1- \frac{2(m_0+ P(\rho) + Q(\rho))}{\rho^{n-2}} & \geq 1-\frac{2(m_0+ \delta m_0 + \epsilon_2)}{\rho^{n-2}}.  
    \end{align}
    When 
    \begin{align}\label{inv nec}
        \rho^{n-2} - 2(m_0 + \delta m_0 + \epsilon_2) > 0
    \end{align}
    we may invert both sides to get that the radial metric coefficient is estimated by
     \begin{align}\label{ineq 1}
        \left(1- \frac{2(m_0+ P(\rho) + Q(\rho))}{\rho^{n-2}} \right)^{-1} &\leq \left(1-\frac{2(m_0+ \delta m_0 + \epsilon_2)}{\rho^{n-2}}\right)^{-1}.
    \end{align} 
    
    Let $c_1(\rho) = \left( 1 - \frac{2 (m_0 + \delta m_0 + \epsilon_2 )}{\rho^{n-2}}  \right)^{-1}\left( 1 - \frac{2 m_0 }{\rho^{n-2}}  \right)$. We compute that 
    \begin{align}\label{eq: c1 comp}
        1   \leq c_1(\rho) = \frac{\rho^{n-2}- 2m_0}{\rho^{n-2} - 2(m_0 + \delta m_0 + \epsilon_2)} 
    \end{align}
    when \eqref{inv nec} holds. It follows that  
    \begin{align}
        c_1(\rho) g_S = \left( 1 - \frac{2 (m_0 + \delta m_0 + \epsilon_2 )}{\rho^{n-2}}  \right)^{-1} d\rho^2 + c_1(\rho) \rho^2 d\Omega^2 \geq \bg.
    \end{align}

    We finish the proof of the first metric estimate by showing there exists $\epsilon_2,\delta m_0$ small enough so that $c_1(\rho) \leq (1+ \epsilon_3)^2$ by finding an upper bound for $c_1(\rho)$ and showing this bound can be controlled as needed. Recall that $\rho_0^{n-2} = 2m_0$ and set
    \begin{align}\label{eq: xi}
        \xi(\epsilon_1) := (\rho_0+\epsilon_1)^{n-2}-2m_0= \sum_{k=0}^{n-3}\binom{n-2}{k} \epsilon_1^{n-2-k}\rho_0^k.
    \end{align}
    
    Note that $c_1(\rho)$ is decreasing in $\rho$ so that the right hand side of \eqref{eq: c1 comp} will be bounded above on $W$ by $c_1(\rho_0 + \epsilon_1)$. We have therefore that
    \begin{align}
        c_1(\rho) &\leq c_1(\rho_0 + \epsilon_1) = \frac{\xi(\epsilon_1)}{\xi(\epsilon_1) - 2\delta m_0 - 2\epsilon_2}.
    \end{align}
    A brief computation shows that $\frac{\xi(\epsilon_1)}{\xi(\epsilon_1) - 2\delta m_0 - 2\epsilon_2} \leq (1+\epsilon_3)^2$ when
    \begin{align}\label{eq: est 2}
        \delta m_0 + \epsilon_2 \leq \frac{\xi(\epsilon_1)(2\epsilon_3 + \epsilon_3^2 ) }{2(1+\epsilon_3)^2}.
    \end{align}
    Note that if \eqref{eq: est 2} holds then \eqref{inv nec} holds, so all inversions and calculations are valid given \eqref{eq: est 2}. It follows that
    \begin{align}
        \bg \leq c_1(\rho) g_S \leq (1+\epsilon_3)^2 g_s
    \end{align}
    on $W$ when $\delta m_0$ and $\epsilon_2$ satisfy \eqref{eq: est 2}.

    We follow the same idea to show that for $\epsilon_2,\delta m_0$ small enough we have
    \begin{align}
        g_S \leq c_2(\rho) \bg \leq (1+\epsilon_3)^2 \bg.
    \end{align}
    First, we study the Schwarzschild metric. We have by definition of the set $W$ that $Q(\rho) + \epsilon_2 \geq 0$ on $W$. 
    Therefore, we have
    \begin{align}
        0 \leq 2(P(\rho)+Q(\rho)+\epsilon_2)\leq 2\delta m_0 + 4\epsilon_2
    \end{align}
    so that when
    \begin{align}\label{eq: est 4}
        2\delta m_0 + 4 \epsilon_2 < \xi(\epsilon_1)
    \end{align}
    we have 
    \begin{align}
        \rho^{n-2} - 2m_0 - 2(P(\rho)+Q(\rho)+\epsilon_2) \geq \xi(\epsilon_1) - 2(P(\rho)+Q(\rho)+\epsilon_2) > 0
    \end{align}
    and
    \begin{align}
        1 - \frac{2 (m_0 + P(\rho) + Q(\rho) + \epsilon_2)}{\rho^{n-2}}  > 0.
    \end{align}
    
   We then have
    \begin{align}\label{eq: that one}
    \begin{split}
        g_S &= \left( 1 - \frac{2 m_0}{\rho^{n-2}}  \right)^{-1} d\rho^2 + \rho^2 d\Omega^2  \\
        &\leq \left( 1 - \frac{2 (m_0 + P(\rho) + Q(\rho) + \epsilon_2)}{\rho^{n-2}}  \right)^{-1} d\rho^2 + \rho^2 d\Omega^2 \\
        & = \left( 1 - \frac{2 (m(\rho) + \epsilon_2)}{\rho^{n-2}}  \right)^{-1} d\rho^2 + \rho^2 d\Omega^2 
    \end{split}
    \end{align}
    where the second inequality follows because $0 \leq P(\rho) + Q(\rho) + \epsilon_2$. Therefore, we choose 
    \begin{align}
        c_2(\rho) = \left( 1 - \frac{2 (m(\rho) + \epsilon_2)}{\rho^{n-2}}  \right)^{-1}\left( 1 - \frac{2 m(\rho)}{\rho^{n-2}}  \right)
    \end{align}
    and can check that $c_2(\rho) \geq 1$ when
    \begin{align}
        \rho^{n-2} - 2(m(\rho) + \epsilon_2) > 0
    \end{align}
    which we have already guaranteed by \eqref{eq: est 4}. Now, we have
    \begin{align}\label{eq: this one}
        c_2(\rho) \bg = \left( 1 - \frac{2 (m(\rho) + \epsilon_2)}{\rho}  \right)^{-1} d\rho^2 + c_2(\rho)\rho^2 d\Omega^2 \geq g_S
    \end{align}
    by the last line of \eqref{eq: that one} and that $c_2(\rho) \geq 1$. 
    
    To complete the proof, we check the conditions under which $c_2(\rho) \leq (1+\epsilon_3)^2$. First, we will find an upper bound for $c_2(\rho)$ in terms of $\xi(\epsilon_1)$, $\epsilon_2$, and $\delta m_0$. Observe that
    \begin{align}
        \begin{split}
            c_2(\rho) &= \frac{\rho^{n-2}-2m(\rho)}{\rho^{n-2}-2m(\rho)-2\epsilon_2} \\
            &= \frac{1}{1-\frac{2\epsilon_2}{\rho^{n-2}-2m(\rho)}}
        \end{split}
    \end{align}
    from which we can see that we may find an upper bound for $c_2(\rho)$ by finding a lower bound for $\rho^{n-2}-2m(\rho)$. Indeed, we have that $\rho^{n-2}-2m(\rho) \geq \xi(\epsilon_1) - 2(\delta m_0 + \epsilon_2)$ so that
    \begin{align}
        \begin{split}
            c_2(\rho) &\leq \frac{1}{1-\frac{2\epsilon_2}{\xi(\epsilon_1)-2(\delta m_0+\epsilon_2)}} \\
            &\leq \frac{1}{1-\frac{2\epsilon_2+ 2\delta m_0}{\xi(\epsilon_1)-2(\delta m_0+\epsilon_2)}}
        \end{split}
    \end{align}
    where in the last line we have once more increased the upper bound by subtracting $2\delta m_0/(\xi(\epsilon_1)-2(\delta m_0+\epsilon_2))$ from the denominator, which is valid when
    \begin{align}\label{eq: est 6}
        4(\delta m_0 + \epsilon_2) < \xi(\epsilon_1).
    \end{align}
    Now we check when 
    \begin{align}
        \begin{split}
            c_2(\rho)\leq \frac{1}{1-\frac{2(\epsilon_2+ \delta m_0)}{\xi(\epsilon_1)-2(\delta m_0+\epsilon_2)}} < (1+\epsilon_3)^2.
        \end{split}
    \end{align}
    First, we get that this holds when 
    \begin{align}
        \frac{2(\epsilon_2+ \delta m_0)}{\xi(\epsilon_1)-2(\delta m_0+\epsilon_2)} < \frac{\epsilon_3^2+2\epsilon_3}{(1+\epsilon_3)^2}.
    \end{align}
    Then we solve for $2(\epsilon_2+\delta m_0)$ to get that we have the inequality when
    \begin{align}
        \begin{split}\label{eq: est 5}
            \delta m_0 + \epsilon_2 < \frac{\xi(\epsilon_1)\left( \epsilon_3^2+2\epsilon_3\right)}{2((1+\epsilon_3)^2+\epsilon_3^2+2\epsilon_3)}.
        \end{split}
    \end{align}
    
    Note that \eqref{eq: est 2}, \eqref{eq: est 6}, and \eqref{eq: est 5} hold when we have
    \begin{align}\label{final est}
        \delta m_0 + \epsilon_2 < \frac{\xi(\epsilon_1)\left( \epsilon_3^2+2\epsilon_3\right)}{4((1+\epsilon_3)^2+\epsilon_3^2+2\epsilon_3)}
    \end{align}
    Thus, we have both metric estimates \eqref{eq: metric estimate} on $W$ when $\epsilon_2$ and $\delta $ satisfy \eqref{final est}.
\end{proof}
As a consequence, we get the desired comparison to the $g_0$ metric.
\begin{cor}\label{cor: met comp 2}
    Suppose the initial data $(M,g,k)$ are smooth, asymptotically flat, and spherically symmetric and satisfy the dominant energy condition \eqref{assumption: DEC} and bounded expansion condition \eqref{assumption: bounded expansions.}. Let $\epsilon_1$, $\epsilon_3 > 0$. Then, for $\epsilon_2(\epsilon_1,\epsilon_3)$ and $\delta(\epsilon_1,\epsilon_2,\epsilon_3)$ small enough, we have 
    \begin{align}\label{eq: metric estimate2}
        \bg \leq (1+\epsilon_3)^{2} g_0 \mbox{ and } g_S \leq (1+\epsilon_3)^{2} g_0
    \end{align}
    on $W$.
\end{cor}
\begin{proof}
    At any point $x \in W$ we have that $g_0 = \bg$ or $g_0 = g_S$. In either case, both estimates hold by Lemma \ref{lem: metric estimate} when $\delta$ and $\epsilon_2$ are chosen thusly with respect to $\epsilon_1$ and $\epsilon_3$.
\end{proof}

\subsection{Volume of the Bad Region}
Now that we have validated the choices of $W_{\epsilon_2}$ and $B_{\epsilon_2}$ by expressing the metric estimate on $W_{\epsilon_2}$, we need to estimate from above the volume of the sets $B_{\epsilon_2}$ and $\widetilde{B}_{\epsilon_2}$. This estimate relies on the bound on the $L^2$ norm of $Q$ given by \eqref{eq: bound on q}.

\begin{lemma} \label{lem: vol out}
    Suppose the initial data $(M,g,k)$ are smooth, asymptotically flat, and spherically symmetric and satisfy the dominant energy condition \eqref{assumption: DEC} and bounded expansion condition \eqref{assumption: bounded expansions.}. Then, the volumes outside the diffeomorphic subregions, denoted by $B_{\epsilon_2}$ and $\widetilde{B}_{\epsilon_2}$, may be estimated by
	\begin{align}\label{est: vol diff 1}
		\mbox{Vol}_{\bg}(B_{\epsilon_2}) \leq\frac{ \delta \omega_{n-1} m_0 A^{2n-1}C}{\epsilon_2^2},
	\end{align}
	and
	\begin{align}\label{est: vol diff 2}
		\mbox{Vol}_{g_S}(\widetilde{B}_{\epsilon_2}) \leq \frac{\delta \omega_{n-1} m_0 A^n C^2\sqrt{(\rho_0+\epsilon_1)^{n-2}}}{\sqrt{\xi(\epsilon_1)} \epsilon_2^2}.
	\end{align}
	In particular, for small enough $\epsilon_1$
	\begin{align}\label{both vol diff}
		\mbox{Vol}_{\bg}(B_{\epsilon_2}), \mbox{Vol}_{g_S}(\widetilde{B}_{\epsilon_2}) \leq \frac{\delta \omega_{n-1}  m_0 A^n C^2\sqrt{(\rho_0+\epsilon_1)^{n-2}}}{\sqrt{\xi(\epsilon_1)} \epsilon_2^2}.
	\end{align}
\end{lemma}
\begin{proof}
    From \eqref{eq: bound on q} and the definition of $B_{\epsilon_2}$,
    \begin{align}
    \begin{split}
        \delta \omega_{n-1} m_0 A^{2n-1}C &\geq \int_{\uep} |Q|^2 d\omega_{\bg} \\
        &\geq \int_{B_{\epsilon_2}} \epsilon_2^2 d\omega_{\bg}
    \end{split}
    \end{align}
    from which \eqref{est: vol diff 1} follows by dividing by $\epsilon_2^2$. To get \eqref{est: vol diff 2}, we calculate that
    \begin{align}
    \begin{split}
        \mbox{Vol}_{\bg}(B_{\epsilon_2}) &= \int_{B_{\epsilon_2}}  d\omega_{\bg} \\
        &= \int_{\widetilde{B}_{\epsilon_2}} \frac{\phi_S}{\phi}  d\omega_{g_S} \\
        &\geq \frac{\phi_S(a(\epsilon_1))}{AC} \mbox{Vol}_{g_S}(\widetilde{B}_{\epsilon_2})
    \end{split}
    \end{align}
    where in the last line we have used the upper bound on $\phi$ given by \eqref{eq: bound on phi} in Lemma \ref{lem: bound on phi}. By \eqref{est: vol diff 1}, we then have that
    \begin{align}
        \mbox{Vol}_{g_S}(\widetilde{B}_{\epsilon_2}) \leq \frac{AC}{\phi_S(a(\epsilon_1))}\frac{ \delta \omega_{n-1} m_0 A^{2n-1}C}{\epsilon_2^2}.
    \end{align}
    Note that $a(\epsilon_1) \geq \rho_0 + \epsilon_1$ so that $\phi_S(a(\epsilon_1)) \geq \phi_S(\rho_0 + \epsilon_1) = \sqrt{\frac{\xi(\epsilon_1)}{(\rho_0+\epsilon_1)^{n-2}}}$ where $\xi(\epsilon_1)$ is defined in \eqref{eq: xi}. The result follows.
\end{proof}

\subsection{Distance Estimate}

Now we consider our intermediary metric space $(U_0,g_0)$ to compute the estimate \eqref{est: dist} for $g_S$ and $\bg$. We wish to obtain some estimates
\begin{align}
\begin{split}
    d_S(x,y) - d_0(x,y) &\leq 2\alpha\\
    d_{\bg}(x,y) - d_0(x,y) &\leq 2\alpha.
    \end{split}
\end{align}
In this section, we compute $\alpha$. First, we need a lemma estimating radial distances in the bad region in the $\bg$ and $g_S$ metrics.

\begin{lemma}\label{lem: ests on B}
    Suppose the initial data $(M,g,k)$ are smooth, asymptotically flat, and spherically symmetric and satisfy the dominant energy condition \eqref{assumption: DEC} and bounded expansion condition \eqref{assumption: bounded expansions.}. Let $(a_k,b_k)$ denote the radii which correspond to the inner and outer boundaries of the components of $B$. Then
    \begin{align}
        \sum_k \int_{a_k}^{b_k} \frac{1}{\phi} d\rho, \sum_k \int_{a_k}^{b_k} \frac{1}{\phi_S} d\rho &\leq \frac{\delta  m_0 A^n C^2}{\sqrt{(\rho_0+\epsilon_1)^{n}\xi(\epsilon_1)} \epsilon_2^2}.
    \end{align}
\end{lemma}
\begin{proof}
    By \eqref{both vol diff} we have
    \begin{align}
		\sum_k\int_{a_k}^{b_k}\int_{S_\rho}d\omega_{\bg}  \leq\frac{\delta \omega_{n-1}  m_0 A^n C^2\sqrt{(\rho_0+\epsilon_1)^{n-2}}}{\sqrt{\xi(\epsilon_1)} \epsilon_2^2}.
	\end{align}
    We can minimize $\rho$ by $\rho_0+\epsilon_1$ to get that 
    \begin{align}
		(\rho_0+\epsilon_1)^{n-1}\omega_{n-1}\sum_k\int_{a_k}^{b_k}\frac{1}{\phi} d\rho  \leq\frac{\delta \omega_{n-1}  m_0 A^n C^2\sqrt{(\rho_0+\epsilon_1)^{n-2}}}{\sqrt{\xi(\epsilon_1)} \epsilon_2^2}.
	\end{align}
    from which the estimate follows. The estimate for the Schwarzschild distance follows in the same way.
\end{proof}

Before we compute the $\alpha$ parameter, we need another result concerning the nature of geodesics in spherically symmetric annuli which may be foliated by spheres of positive mean curvature. We say that a curve \begin{align}
c: [t_1,t_2] \to M, \quad c(t) = (r(t),\theta(t))
\end{align}
on such a manifold has a radial relative maximum on an interval $[a,b] \subsetneq [t_1,t_2]$ if $r(t)$ is constant on $[a,b]$ and there exists $\delta$ small enough such that $r(t)$ is strictly increasing on $[a - \delta,a)$ and strictly decreasing on $(b,b+\delta]$. We prove that a geodesic on such a manifold may not have a radial relative maximum. Note that the relative radial maximum may only occur on the outer boundary of the annulus and not the inner boundary. This is why we chose the metric $g_0$ to be smooth up to the outer boundary of each annulus and continuous up to the inner boundary.

\begin{lemma}\label{lem: no max}
    Let $V$ be a spherically symmetric annulus with boundary with Riemannian metric $\hat{g}$ which is smooth up to the outer boundary, smooth in the interior, and continuous up to the inner bondary. Let $V$ be foliated by spheres of positive mean curvature. Then for points $x,y \in U$, the distance minimizing geodesic which realizes $d_{\hat{g}}(x,y) = \ell_{\hat{g}}(\gamma)$ may not achieve a radial relative maximum.
\end{lemma}
\begin{proof}
    We write our spherically symmetric metric in arc length coordinates as 
    \begin{align}
        \hat{g} = ds^2 + \rho^2(s)g_S.
    \end{align}
    Suppose $\gamma(t)$ is a distance minimizing geodesic between two points $x$ and $y$. Suppose further that $\gamma(t)$ has a radial relative maximum on $[a,b]$ at a radius $s_1$ which might be on the outermost boundary of the annulus. Let $\delta_1$, $\delta_2$ be such that $\gamma(a - \delta_1) = (s_2,\theta_1)$ and $\gamma(b+\delta_2) = (s_2,\theta_2)$. Further, choose $\delta_1$ and $\delta_2$ small enough enough so that $\gamma:(a-\delta_1,b+\delta_2) \to M$ lies above $S_{s_2}$. Such a choice of $\delta_j$ is always possible because $[a,b]$ is a radial relative maximum. Let $s \circ \gamma: [a,b] \to \IR_+$ be the distance between $\gamma$ and $S_{s_2}$.

    We may compute the mean curvature as 
    \begin{align}
        0 < H = \frac{(n-1)\rho_{,s}}{\rho}
    \end{align}
    and the second fundamental form is
    \begin{align}
        2A = 2\pd_s(\rho^2(s)g_S) = \rho \rho_{,s}g_S 
    \end{align}
    which is positive because positive mean curvature implies $\rho_{,s}$ is positive. Further, we write
    \begin{align}
        A = \nabla_{S}^2 s
    \end{align}
    because, as a distance function, $|\nabla s| = 1$. We express the Hessian of $s$ as
    \begin{align}
        \mbox{Hess}(s) = \nabla^2 s (\pd_s,\pd_s) + A
    \end{align}
    where
    \begin{align}
        \nabla^2s(\pd_s,\pd_s) = \bracket{\nabla_{\pd_s}s,\pd_s} = 0.
    \end{align}

    Note that $\mbox{Hess}(s)(v,v)$ may be zero if and only if $v$ is strictly radial. Recall that $\Dot{\gamma}(t)$ is tangent to $S_{s_1}$ on $[a,b]$. As $\Dot{\gamma}(t)$ varies smoothly in a smooth metric, we may choose $\delta$ small enough so that $\Dot{\gamma}(t)$ may not have a strictly radial component on $(a-\delta,a+\delta)$. It follows that 
    \begin{align}
        \mbox{Hess}(s)(\Dot{\gamma}(t),\Dot{\gamma}(t)) > 0
    \end{align}
    on that section, which contradicts that $s \circ \gamma$ is concave down or flat on that interval; ie $[a,b]$ cannot be a relative radial maximum for $\gamma$.
\end{proof}

A consequence of this same argument is that a geodesic may not be radially constant except on the innermost boundary. In other words, there are no radial shortcuts except at the inner boundary. With this result in hand, we estimate the difference in distances directly. This is the result in which we need the specific definition of $g_0$, which is smooth up to the outer boundary of each annulus and continuous on the inner boundary.

\begin{lemma}\label{lemm: diff in dist}
	Suppose the initial data $(M,g,k)$ are smooth, asymptotically flat, and spherically symmetric and satisfy the dominant energy condition \eqref{assumption: DEC} and bounded expansion condition \eqref{assumption: bounded expansions.}. Then we can compute that the difference in distances with respect to the $g_S, \bg,$ and $g_0$ metrics for $x,y \in W$ satisfies 
	\begin{align}
    \begin{split}
		d_{S}(x,y)- d_0(x,y)  &\leq 2\alpha  \\
        d_{\bg}(x,y) - d_0(x,y) & \leq 2\alpha
    \end{split}
	\end{align}
    where 
    \begin{align}
        \alpha = 2\epsilon_3 \pi D_0 + \frac{\delta  m_0 A^n C^2}{\sqrt{(\rho_0+\epsilon_1)^{n}\xi(\epsilon_1)} \epsilon_2^2}
    \end{align}
\end{lemma}
\begin{proof}
    For $x$ and $y \in W$, let $\gamma^0$ by a geodesic connecting them in the $g_0$ metric, and $\gamma^S$ a geodesic in the Schwarzschild metric $g_S$. Let $\ell_{g_0}$ and $\ell_{g_S}$ denote the lengths in each metric.
    
    To begin, we have
    \begin{align}
        d_{g_S}(x,y) = \ell_{g_S}(\gamma^S) \leq \ell_{g_S}(\gamma^0) 
    \end{align}
    because $\gamma^S$ and $\gamma^0$ both connect points the points $x$ and $y$ in $W$. Now we break up $\gamma^0$ into the disjoint sets $\gamma^0 \cap W$ and $\gamma^0 \cap B$ and estimate $\gamma^0 \cap W$ in terms of $g_0$ using the metric comparison Corollary \ref{cor: met comp 2}:  
    \begin{align}
    \begin{split}
        d_{g_S}(x,y) & \leq \ell_{g_S}(\gamma^0\cap W) + \ell_{g_S}(\gamma^0 \cap B)\\
        &\leq (1+\epsilon_3)\ell_{g_0}(\gamma^0\cap W) +  \ell_{g_S}(\gamma^0 \cap B).
    \end{split}
    \end{align}
    If we express the distance between $x$ and $y$ in $g_0$ as
    \begin{align}
        d_{g_0}(x,y) = \ell_{g_0}(\gamma^0 \cap W) + \ell_{g_0}(\gamma^0 \cap B)
    \end{align}
    we get that
    \begin{align}\label{eq: dist1}
        d_{g_S}(x,y) - d_{g_0}(x,y) \leq \epsilon_3 \ell_{g_0}(\gamma^0 \cap W) + \ell_{g_S}(\gamma^0 \cap B) - \ell_{g_0}(\gamma^0 \cap B).
    \end{align}
    Repeating the procedure with $\bg$ and $g_0$ gives 
    \begin{align}\label{eq: dist2}
        d_{\bg}(x,y) - d_{g_0}(x,y) \leq \epsilon_3 \ell_{g_0}(\gamma^0 \cap W) + \ell_{\bg}(\gamma^0 \cap B) - \ell_{g_0}(\gamma^0 \cap B).
    \end{align}
    If we can estimate the right hand side of \eqref{eq: dist1} and \eqref{eq: dist2} by the same constant independent of $x$ and $y$ we will have $\alpha$.
    
    By Lemma \ref{lemm diams}, $\ell_{g_0}(\gamma^0 \cap W) \leq 4\pi D_0$. It remains only to estimate the differences
    \begin{align}
        \ell_{g_S}(\gamma^0 \cap B) - \ell_{g_0}(\gamma^0 \cap B) \quad \mbox{ and } \quad \ell_{\bg}(\gamma^0 \cap B) - \ell_{g_0}(\gamma^0 \cap B).
    \end{align}

    By Lemma \ref{lem: no max}, for each annular region in which $g_0 = \bg$ or $g_0 = g_S$, $\gamma_0$ may not achieve a relative maximum either on the interior of the region or on the outermost boundary. It follows that we can split $\gamma^0$ into at two disjoint components: one which is radially decreasing and one which is radially increasing. 

    Let $x = (\rho_1,\theta_1)$, $y = (\rho_2,\theta_2)$ and, without loss of generality, let $\gamma_1^0(t) = (\rho_1 - t,\theta_1(t))$ be radially decreasing for $t \in [0, t_3]$ and $\gamma_2^0(t) = (\rho_{min}+t,\theta_2(t))$ be radially increasing for $t \in [t_3, \rho_2-\rho_{min}]$. We may compute
    \begin{align}
    \begin{split}
        \ell_{g_S}(\gamma^0 \cap B) - \ell_{g_0}(\gamma^0 \cap B) &= \int_{\gamma_1^0\cap B}\sqrt{\frac{1}{\phi_S^2} + \rho^2\beta_{n-1}} - \sqrt{\frac{1}{\phi_0^2} + \rho^2\beta_{n-1} }dt \\
        &\quad + \int_{\gamma_2^0\cap B}\sqrt{\frac{1}{\phi_S^2} + \rho^2\beta_{n-1}} - \sqrt{\frac{1}{\phi_0^2} + \rho^2\beta_{n-1} }dt
        \end{split}
    \end{align}
    where $\beta_{n-1}$ is the distance on the $(n-1)$-sphere and we abuse notation slightly so $\gamma_1^0 \cap B$ is the set of $t$ values for which $\gamma_1^0$ lies in $B$. Then,
    \begin{align}
    \begin{split}
        \ell_{g_S}(\gamma_1^0 \cap B) - \ell_{g_0}(\gamma_1^0 \cap B) &= \int_{\gamma_1^0\cap B}\frac{\frac{1}{\phi_S^2}-\frac{1}{\phi_0^2}}{\sqrt{\frac{1}{\phi_S^2} + \beta_{n-1}} + \sqrt{\frac{1}{\phi_0^2} + \beta_{n-1} }}dt \\
        &\leq \int_{\gamma_1^0\cap B}\frac{\frac{1}{\phi_S^2}-\frac{1}{\phi_0^2}}{\sqrt{\frac{1}{\phi_S^2} + \frac{1}{\phi_0^2} }}dt\\
        &\leq \int_{\gamma_1^0\cap B}\frac{\frac{1}{\phi_S^2}-\frac{1}{\phi_0^2}}{\sqrt{\frac{1}{\phi_S^2} - \frac{1}{\phi_0^2} }}dt \\
        &\leq \int_{\gamma_1^0\cap B}\sqrt{\frac{1}{\phi_S^2} - \frac{1}{\phi_0^2} } dt\\
        & \leq \int_{\gamma_1^0\cap B} \frac{1}{\phi_S } \\
        &\leq  \sum_k \int_{a_k}^{b_k} \frac{1}{\phi_S } d\rho
    \end{split}
    \end{align}
    where the last estimate holds for any possible choice of $x$ and $y$ because, by virtue of being monotonic in $\rho$, $\gamma_1^0$ may be in each component of $B$ at most once. The same estimates hold for $\gamma_2^0$, so by Proposition \ref{lem: ests on B} we have now that
    \begin{align}
        \ell_{g_S}(\gamma^0 \cap B) - \ell_{g_0}(\gamma^0 \cap B) \leq 2 \frac{\delta  m_0 A^n C^2}{\sqrt{(\rho_0+\epsilon_1)^{n}\xi(\epsilon_1)} \epsilon_2^2}.
    \end{align}

    The same procedure for $\bg$ gives that 
    \begin{align}
        \ell_{\bg}(\gamma^0 \cap B) - \ell_{g_0}(\gamma^0 \cap B) \leq 2 \frac{\delta  m_0 A^n C^2}{\sqrt{(\rho_0+\epsilon_1)^{n}\xi(\epsilon_1)} \epsilon_2^2}
    \end{align}
    and we have
    \begin{align}
    \begin{split}
        d_S(x,y) - d_0(x,y) &\leq 2\left(2\epsilon_3 \pi D_0 + \frac{\delta  m_0 A^n C^2}{\sqrt{(\rho_0+\epsilon_1)^{n}\xi(\epsilon_1)} \epsilon_2^2} \right)\\
        d_{\bg}(x,y) - d_0(x,y) &\leq 2\left(2\epsilon_3 \pi D_0 + \frac{\delta  m_0 A^n C^2}{\sqrt{(\rho_0+\epsilon_1)^{n}\xi(\epsilon_1)} \epsilon_2^2}\right)
        \end{split}
    \end{align}
    for all $x,y \in W$, so we may choose 
    \begin{align}
        \alpha = 2\epsilon_3 \pi D_0 + \frac{\delta  m_0 A^n C^2}{\sqrt{(\rho_0+\epsilon_1)^{n}\xi(\epsilon_1)} \epsilon_2^2} .
    \end{align}
\end{proof}

\subsection{Proof of Proposition}

We may now apply Theorem \ref{vadb bdy} to prove Proposition \ref{prop: inter metric space}. 

\begin{proof}[Proof of Proposition 2]
    Let $\epsilon > 0$. First, consider $d_{\mf}(U_0,\uep)$. The map $F: \uep \to U_0$ which identifies $(\rho, \theta) \in \uep$ with $(\rho, \theta) \in U_0$ is a smooth diffeomorphism and so is biLipschitz with smooth inverse. Also, by definition of $U_0$, this map is distance nonincreasing. If we take $\alpha$ as in Lemma \ref{lemm: diff in dist}, we have $d_{\bg}(x,y) \leq d_0(x,y)+2\alpha$ for all $x,y \in W$. 

    It follows by Theorem \ref{vadb bdy} and Lemmas \ref{lemm: volumes of regions}, \ref{lem: bdy areas}, \ref{lemm diams}, \ref{lem: vol out}, and \ref{lemm: diff in dist} that
    \begin{align}
        d_{\mf}(U_0,\uep) \leq 2\left( \frac{\delta \omega_{n-1}  m_0 A^n C^2\sqrt{(\rho_0+\epsilon_1)^{n-2}}}{\sqrt{\xi(\epsilon_1)} \epsilon_2^2} \right) + h(\omega_{n-1}A \widetilde{D}) + h(\omega_{n-1}(A^{n-1}+a(\epsilon_1)^{n-1}))
    \end{align}
    where 
    \begin{align}
        h= \sqrt{8\alpha \pi D_0 + \alpha^2 }.
    \end{align}
    and $\alpha = 2\epsilon_3 \pi D_0 + \frac{\delta  m_0 A^n C^2}{\sqrt{(\rho_0+\epsilon_1)^{n}\xi(\epsilon_1)} \epsilon_2^2}$. For each $\epsilon_1, \epsilon > 0$ there exists $\epsilon_3$, $\epsilon_2$, and $\delta$ small enough so that each term on the right hand side is less than $\epsilon/9$ and we have
    $$d_{\mf}(U_0,\uep) < \epsilon/3.$$
    The proof for $\wuep$ follows in the same way.
\end{proof}

\begin{rem}
    In the proof of Proposition 2, we first have $\epsilon_1$, $\epsilon$ small, then choose $\epsilon_3$ sufficiently small so that $2\epsilon_3 \pi D_0$ is small. We then must choose $\epsilon_2$ and $\delta$ small enough to satisfy Corollary \ref{cor: met comp 2}. This, as well as the fact that $\xi(\epsilon_1)$ has largest term $\rho_0^{n-3}\epsilon_1$, means we must choose $\delta$ much smaller than $\rho_0^{n-1}\epsilon_1\epsilon_2^2$ so that $\frac{\delta \omega_{n-1}  m_0 A^n C^2\sqrt{(\rho_0+\epsilon_1)^{n-2}}}{\sqrt{\xi(\epsilon_1)} \epsilon_2^2}$ and $\frac{\delta  m_0 A^n C^2}{\sqrt{(\rho_0+\epsilon_1)^{n}\xi(\epsilon_1)} \epsilon_2^2}$ are small.
\end{rem}

\section{Proof of Theorem}\label{proof}

We prove first the estimate on the difference in volumes, then the convergence of the second fundamental forms in separate arguments.

\begin{prop}\label{prop vol diff}
	Suppose the initial data $(M,g,k)$ are smooth, asymptotically flat, and spherically symmetric and satisfy the dominant energy condition \eqref{assumption: DEC} and bounded expansion condition \eqref{assumption: bounded expansions.}. Then, the difference of total volumes may be estimated by
	\begin{align}
    \begin{split}
		|\mbox{Vol}_{\bg}(\uep) - \mbox{Vol}_{g_S}(\widetilde{U}_{a(\epsilon_1)}^A)| &\leq  \frac{2\delta \omega_{n-1}  m_0 A^n C^2\sqrt{(\rho_0+\epsilon_1)^{n-2}}}{\sqrt{\xi(\epsilon_1)} \epsilon_2^2}  \\
		& \quad + ((1+\epsilon_3)^n - 1) \widetilde{D} A^{n-1}\omega_{n-1}.
  \end{split}
	\end{align}
    In particular for any $\epsilon > 0$, there exists $\epsilon_3, \epsilon_2,\epsilon_1,$ and $\delta$ small enough that
    \begin{align}
        |\mbox{Vol}_{\bg}(\uep) - \mbox{Vol}_{g_S}(\widetilde{U}_{a(\epsilon_1)}^A)| < \epsilon/4.
    \end{align}
\end{prop}
\begin{proof}
	We estimate each term on the right hand side of
	\begin{align}
    \begin{split}
		|\mbox{Vol}_{\bg_j}(\uep) - \mbox{Vol}_{g_S}(\widetilde{U}_{a(\epsilon_1)}^A)| \leq& |\mbox{Vol}_{\bg}(\uep) - \mbox{Vol}_{\bg}(W_{\epsilon_2})| + |\mbox{Vol}_{\bg}(W_{\epsilon_2}) - \mbox{Vol}_{g_S}(\widetilde{W}_{\epsilon_2})| \\
		&+ |\mbox{Vol}_{g_S}(\widetilde{W}_{\epsilon_2}) - \mbox{Vol}_{g_S}(\widetilde{U}_{a(\epsilon_1)}^A)|.
  \end{split}
	\end{align}
	By Lemma \ref{lem: vol out}, we have
	\begin{align}
		|\mbox{Vol}_{\bg}(\uep) - \mbox{Vol}_{\bg}(W_{\epsilon_2})|, 
		|\mbox{Vol}_{g_S}(\widetilde{W}_{\epsilon_2}) - \mbox{Vol}_{g_S}(\widetilde{U}_{a(\epsilon_1)}^A)| & \leq \frac{\delta \omega_{n-1}  m_0 A^n C^2\sqrt{(\rho_0+\epsilon_1)^{n-2}}}{\sqrt{\xi(\epsilon_1)} \epsilon_2^2}.  
	\end{align}
	By Lemmas \ref{lem: metric estimate} and \ref{lemm: volumes of regions} we have
	\begin{align}
    \begin{split}
		|\mbox{Vol}_{\bg}(W_{\epsilon_2}) - \mbox{Vol}_{g_S}(\widetilde{W}_{\epsilon_2})| &\leq ((1+\epsilon_3)^n - 1) \mbox{Vol}_{g_S}(\widetilde{W}_{\epsilon_2})\\
		 &\leq ((1+\epsilon_3)^n - 1) \widetilde{D} A^{n-1}\omega_{n-1}.
   \end{split}
	\end{align}
\end{proof}

\begin{lemma}\label{lem: conv of 2nd fund form}
    If the initial data $(M,g,k)$ are smooth, asymptotically flat, and spherically symmetric and satisfy the dominant energy condition \eqref{assumption: DEC} and bounded expansion condition \eqref{assumption: bounded expansions.}, then for any $\epsilon > 0$ there is $\delta > 0$ small enough so that $\norm{\sqrt{\phi}(k - \pi)}_{L^2(\Sigma,\bg)} < \epsilon$. 
\end{lemma}
\begin{proof}
    From Theorem \ref{Penrose Inequality}, we have
\begin{align}
\begin{split}
   \delta m_0 = m_{ADM} - m_0 &=  \frac{1}{2\omega_{n-1}(n-1)}\int_\Sigma  \phi \left(16\pi(\mu-J(w)) + |h-K|_{\Sigma}|_{\bg}^2 + 2|q|_{\bg}^2 \right)d \omega_{\bg}  \\
   & \geq \frac{1}{2\omega_{n-1}(n-1)}\int_\Sigma  \phi |h-K|_{\Sigma}|_{\bg}^2 d \omega_{\bg} 
   \end{split}
\end{align}
so that the $L^2$ norm of $\sqrt{\phi}(h-K|_{\Sigma}) = \sqrt{\phi}(k-\pi)$ is controlled by $\delta$ on all of $\Sigma$, from which the second fundamental form convergence follows. 
\end{proof}

Now we may prove the theorem.
\begin{proof}[Proof of Theorem 1]

By Lemma \ref{lem: conv of phi}  we have the $L_{loc}^2$ convergence of $\phi$ to $\phi_S$ and by Lemma \ref{lem: conv of 2nd fund form} we have the $L^2$ convergence of the second fundamental forms. It remains to show that

\begin{align}
        d_{\mv\mf}(\uep, \widetilde{U}_{a(\epsilon_1)}^A) < \epsilon
    \end{align}
for small enough $\delta$. 

By the triangle inequality, we may write
\begin{align}
\begin{split}
    d_{\mv\mf}(\uep, \widetilde{U}_{a(\epsilon_1)}^A) &= d_{\mf}(\uep, \widetilde{U}_{a(\epsilon_1)}^A) + |\mbox{Vol}_{\bg}(\uep) - \mbox{Vol}_{g_S}(\widetilde{U}_{a(\epsilon_1)}^A)|  \\
    &\leq d_{\mf}(\uep, U_0) + d_{\mf}(\wuep,U_0)  \\
    &\quad + |\mbox{Vol}_{\bg}(\uep) - \mbox{Vol}_{g_S}(\widetilde{U}_{a(\epsilon_1)}^A)| .
    \end{split}
\end{align}
By Propositions \ref{prop: size of inner}, \ref{prop: inter metric space}, and \ref{prop vol diff} we have for small enough $\epsilon_3$, $\epsilon_2$, and $\delta$ that
\begin{equation}
\begin{aligned}
    d_{\mf}(\uep, U_0) &< \epsilon/3 \\
    d_{\mf}(\wuep,U_0) &< \epsilon/3 \\
    |\mbox{Vol}_{\bg}(\uep) - \mbox{Vol}_{g_S}(\widetilde{U}_{a(\epsilon_1)}^A)| &< \epsilon/3.
\end{aligned}
\end{equation}
The bound follows.
\end{proof}

\printbibliography

\end{document}